\documentclass{ieeetj}
\usepackage{cite}
\usepackage{amsmath,amssymb,amsfonts}
\usepackage{algorithmic}
\usepackage{graphicx,color}
\usepackage{textcomp}
\usepackage{xcolor}
\usepackage{hyperref}
\hypersetup{hidelinks=true}
\usepackage{algorithm,algorithmic}
\def\BibTeX{{\rm B\kern-.05em{\sc i\kern-.025em b}\kern-.08em
    T\kern-.1667em\lower.7ex\hbox{E}\kern-.125emX}}
\AtBeginDocument{\definecolor{tmlcncolor}{cmyk}{0.93,0.59,0.15,0.02}\definecolor{NavyBlue}{RGB}{0,86,125}}

\usepackage{algorithm}
\usepackage{graphicx}
\usepackage{mathtools}
\usepackage{subfigure}
\usepackage{bm}
\usepackage{bbm}
\usepackage{tabularx}
\usepackage{amsthm}
\usepackage{subcaption}

\renewcommand{\thetable}{\Roman{table}} 
\makeatletter
\renewcommand{\p@subtable}{}                       
\makeatother

\makeatletter
\renewcommand{\p@subfigure}{}                       
\makeatother

\theoremstyle{plain}
\newtheorem{proposition}{Proposition}
\newtheorem{corollary}{Corollary}
\theoremstyle{definition}
\newtheorem{definition}{Definition}

\newtheorem{assump}{Assumption}

\newcommand{\argmin}{\mathop{\mathrm{argmin}}}
\newcommand{\argmax}{\mathop{\mathrm{argmax}}}

\def\OJlogo{\vspace{-4pt}$<$Society logo(s) and publication title will appear here.$>$}
\def\seclogo{\vspace{10pt}$<$Society logo(s) and publication title will appear here.$>$}

\newcommand{\ev}{\mathbf{e}}

\newcommand{\uv}{\mathbf{u}}

\newcommand{\xv}{\mathbf{x}}

\newcommand{\Dm}{\mathbf{D}}

\newcommand{\Wm}{\mathbf{W}}

\newcommand{\Xm}{\mathbf{X}}

\begin{document}
\receiveddate{XX Month, XXXX}
\reviseddate{XX Month, XXXX}
\accepteddate{XX Month, XXXX}
\publisheddate{XX Month, XXXX}
\currentdate{XX Month, XXXX}
\doiinfo{XXXX.2022.1234567}

\markboth{}{Author {et al.}}

\title{Denoising for Neuromorphic Cameras\\
Based on Graph Spectral Features }

\author{\IEEEauthorblockN{Shimpei Harada, Junya Hara,~\IEEEmembership{Member, IEEE}, Hiroshi Higashi~\IEEEmembership{Member, IEEE}, and Yuichi Tanaka,~\IEEEmembership{Senior Member, IEEE}}}
\affil{S. Harada, J. Hara, and Y. Tanaka are with the Graduate School of Engineering, The University of Osaka, Osaka, Japan.}
\affil{H. Higashi is with the Graduate School of Science and Engineering, Kansai University, Osaka, Japan.}
\corresp{Corresponding author: First A. Author (email: s.harada@msp-lab.org).}
\authornote{The preliminary version of this paper is presented in \cite{harada2024denoising}.\\
This work was supported in part by JSPS KAKENHI under Grant 23H01415, 22H05163, and 22K12500, JST AdCORP under Grant JPMJKB2307.}

\begin{abstract}
Neuromorphic cameras, also known as event-based cameras, can detect changes in the environmental brightness asynchronously and independently for each pixel.
They output the brightness changes, i.e., events, as 3-D (2-D pixel coordinates + time) streaming data. 
While event-based cameras are used in many applications because of their desirable characteristics, e.g., high temporal resolution, low latency, low power consumption, and high dynamic range, their measurements contain considerable noise due to their high sensitivity.
In this paper, we propose a denoising method for event-based cameras based on graph spectral features. 
In the proposed method, we first construct a graph where nodes represent events and edges represent the spatiotemporal distance between the events.
To calculate the graph-specified parameter that controls the connectivities of a constructed graph,
we utilize the prior on the density of 3-D events.
We then calculate the eigenvectors of the graph Laplacian.
The obtained eigenvectors are used to extract noiseless events directly.
In the calculation of the eigenvectors, we customize the graph Laplacian to reorder its eigenvalues. This allows us to leverage fast eigensolver algorithms instead of the naive eigendecomposition and thereby reduce computational complexity.
In experiments on synthetic and real-world event data, we demonstrate that the proposed method effectively removes noise events from the raw events compared to alternative methods.
\end{abstract}

\begin{IEEEkeywords}
Event-based camera, Fiedler vector, $\epsilon$-neighbor graph
\end{IEEEkeywords}


\maketitle
\section{Introduction}
\IEEEPARstart{N}{euromorphic} cameras, also known as event-based cameras, are widely utilized in many imaging applications \cite{gallegoEventbasedVisionSurvey2020,leeReviewBioinspiredVision2015}, such as video frame interpolation \cite{tulyakovTimeLensEventbased2022,wuVideoInterpolationEventDriven2022}, object recognition \cite{schaeferAegnnAsynchronousEventbased2022,oconnorRealtimeClassificationSensor2013}, and simultaneous localization and mapping (SLAM) \cite{gallegoEventbased6DOFCamera2017,weikersdorferSimultaneousLocalizationMapping2013}.
They capture a scene from environmental brightness changes
by dynamic vision sensors (DVSs).

DVSs have several desirable characteristics \cite{delbruckActivitydrivenEventbasedVision2010,lichtsteiner128Backslashtimes1282008}: high temporal resolution, low power consumption, low latency, and high dynamic range.
In contrast, conventional frame-based cameras could suffer from capturing fast-moving objects and wide dynamic range scenes.

DVSs detect the environmental brightness changes asynchronously and independently for each pixel.
They output the changes as 3-D (2-D pixel coordinates + time) streaming data called \textit{events} \cite{gallegoEventbasedVisionSurvey2020}. Events are distributed irregularly in the 3-D space. Fig.~\ref{fig:example} shows an example of event stream data. 

\begin{figure}[t]
\centering
\includegraphics[width = 0.7\linewidth]{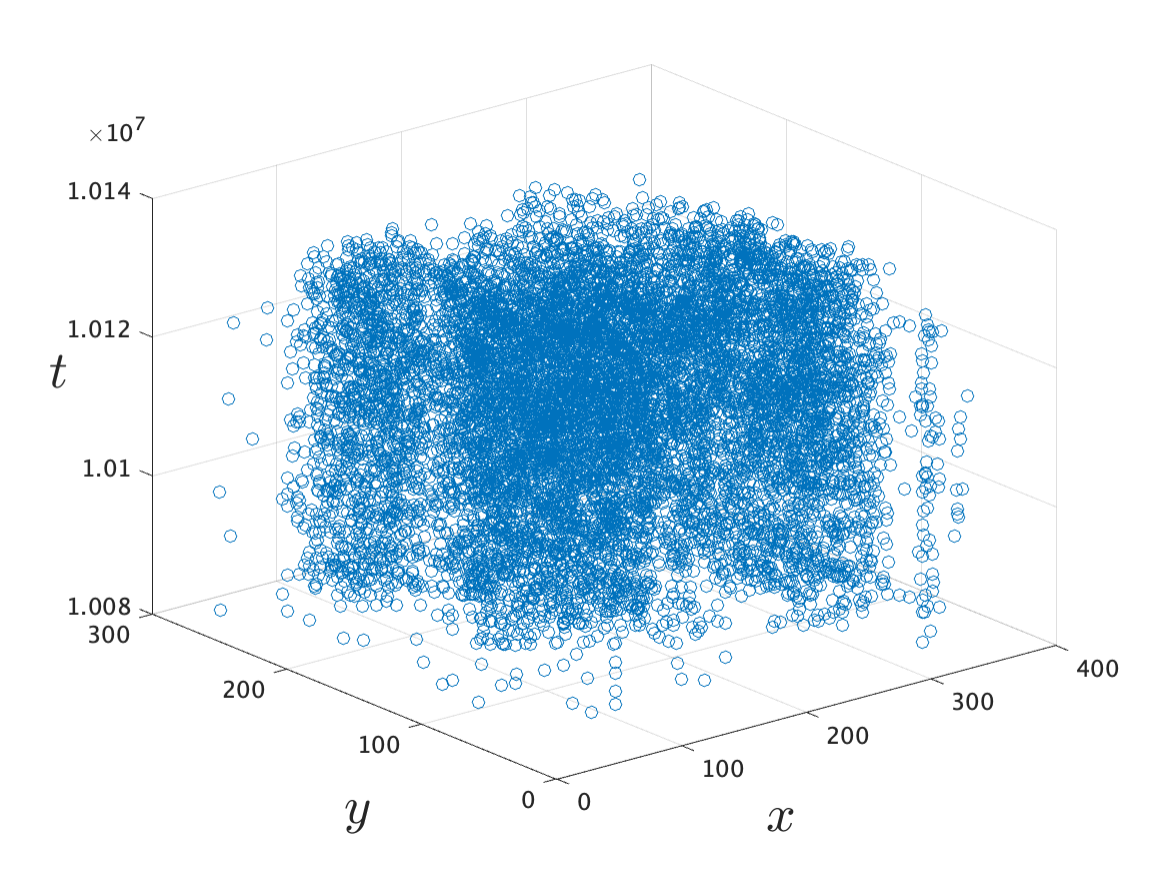} 
\caption{Example of signals captured by an event-based camera. $x$ and $y$ are the axes of the spatial coordinates. $t$ is the axis of the time coordinate.}
\label{fig:example}
\end{figure}

The acquisition process of DVSs usually results in considerable noise caused by overheating pixel (i.e., hot pixel) and background activity (BA) noise due to a trade-off with their high sensitivity \cite{gallegoEventbasedVisionSurvey2020}.
Consequently, numerous noise events are contained in the observed events.
Denoising of event streams, which classifies real and noise events, is crucial for applications using neuromorphic cameras.

For denoising of events, spatial low-pass filters are often applied to a set of events within a short time window
\cite{xieImprovedApproachVisualizing2017,xieDVSImageNoise2018}.
These methods often ignore the temporal correlation of events.
To resolve the problem, methods for extracting events having dense neighbors in the 3-D space have been proposed \cite{fengEventDensityBased2020a,zhangNeuromorphicImagingDensitybased2023}.
However, these methods often mistakenly detect noise events that accidentally have dense neighbors because they perform event denoising based on their \textit{individual} densities.  

In this paper, we propose an effective denoising method for event-based cameras based on graph spectral features.
In contrast to existing methods, we consider the \textit{inter-cluster} densities formed by real events. The key idea is that we identify the clusters of real events using a graph connecting nearby events in the 3-D space, where real events typically form large clusters in the 3-D space while noise events scatter in fragmented pieces. This characteristic is referred to as \textit{density prior}.

Our proposed method is performed in two steps. First, we construct an $\epsilon$-neighbor graph ($\epsilon$NG) from the raw events, where events (nodes) are connected to their neighbors within the distance $\epsilon$ by edges.
Second, we detect large connected components in the $\epsilon$NG based on the graph spectral features.
Specifically, we analyze the eigenvectors of the graph Laplacian associated with the small eigenvalues, where the magnitudes of these eigenvectors indicate the strength of connectivity.

The proposed event denoising, i.e., a large component detection of an event-based graph, is based on spectral graph theory \cite{chungSpectralGraphTheory1997}.
We show that, if real events form large clusters, the largest components correspond to non-zero entries in eigenvectors for the smallest eigenvalues.

To efficiently obtain eigenvectors, we utilize a matrix derived from the graph Laplacian.
Specifically, it can be obtained by reordering the non-zero eigenvalues in the graph Laplacian such that
the \textit{largest} eigenvalues of the matrix correspond to the \textit{smallest non-zero} eigenvalues of the graph Laplacian. This customized graph operator allows us to utilize fast eigensolver algorithms for calculating eigenvectors for small eigenvalues.

We perform denoising experiments with synthetic and real-world datasets.
The proposed method outperforms alternative model- and neural network-based methods in denoising accuracy.

The preliminary work of this paper is presented in \cite{harada2024denoising}.
It focuses on the scenario where real events form a single large cluster.
In this paper, we generalize its theoretical features by extending it to handle scenarios where real events form multiple clusters.
We also examine computational complexity and perform comprehensive experiments.

The outline of this paper is as follows:
Section~\ref{sec:related_work} starts by describing the event acquisition model of neuromorphic cameras. Then, existing studies on event denoising are reviewed. In Section ~\ref{sec:preliminary}, the basics of graph spectral theory and a graph construction method are introduced. We propose the event denoising method in Section~\ref{sec:proposed}. 
We derive the computational complexity of the proposed method to show that its computational effectiveness in Section~\ref{sec:complexity}.
In Section~\ref{sec:experiment}, we evaluate the effectiveness of the proposed method via conducting event denoising experiments.  
Section~\ref{sec:conclusion} concludes this paper.

\textit{Notation:} 
A vector and a matrix are denoted by boldface lower and upper cases, respectively.
The $i$th element of vector $\mathbf{v}$ is denoted by $[\mathbf{v}]_i$.
The $(n,m)$-element of a matrix $\mathbf{A}$ is denoted by $[\mathbf{A}]_{n,m}$.
The $m$th column and $n$th row of $\mathbf{A}$ are denoted by $[\mathbf{A}]_{:,m}$ and $[\mathbf{A}]_{n,:}$, respectively. We denote a $\ell_2$ norm by $\|\cdot\|$. 
The cardinality of a set $\mathcal{A}$ is denoted by 
$|\mathcal{A}|$.
The element-wise product is denoted by $\circ$.

\section{Related Works}\label{sec:related_work}
In this section, we first describe the event acquisition model of neuromorphic cameras. Then, we briefly review existing studies on denoising for event-based cameras.

\subsection{Neuromorphic Camera}
A neuromorphic camera detects environmental brightness changes as events by DVSs \cite{lichtsteiner128Backslashtimes1282008}.
Every DVS, which is associated with 2-D pixel coordinate, outputs events asynchronously and independently of the pixels \cite{gallegoEventbasedVisionSurvey2020}. Let $I(x,y,t)$ be the brightness at the spatial location $(x,y)$ at time instance $t$.
The brightness change is measured by the temporal log-difference as follows:
\begin{equation}
\Delta\ln I(x,y,t)=\ln I(x,y,t)-\ln I(x,y,t-\Delta t),
\end{equation}
where $\Delta t$ is a small perturbation to $t$. The event is triggered if $|\Delta\ln I(x,y,t)| > C$ where $C$ is a controllable contrast threshold. Formally, a set of events is defined as follows\footnote{An event may have another entry $p\in \{-1,1\}$ called polarity \cite{HighSpeedHigh}, where $p$ is determined such that $\Delta\ln I(x,y,t) = pC$. For conciseness, we omit the polarity.}: 
\begin{equation}
\begin{split}
    &\mathcal{X}_{\text{E}}\coloneqq \{\ev=[x,y,t]^\top\mid|\Delta\ln I(x,y,t)|>C,(x,y,t)\in\mathcal{X}\},
\end{split}
\end{equation}
where $\mathcal{X}$ is the feasible set of $(x,y,t)$.
Let $\ev_k$ be the $k$th detected event.
Then, denoting the collection of $\ev_k\in\mathcal{X}_{\text{E}}$ by $\mathbf{E}=[\ev_1,\ldots,\ev_N]$, we can regard $\mathbf{E}$ as 3-D event streaming data.

\subsection{Existing Event Denoising Methods}\label{sub:exist_event_denoise}
DVSs usually yield two different types of noise, hot pixel noise and BA noise.
Hot pixel noise continuously appears along the temporal direction.
It is caused by the overheating of sensors due to prolonged exposure of pixels similar to traditional image sensors.
BA noise occurs from dark and leakage currents in a DVS circuit, which is distributed randomly in space.
Particularly, a dark lightning environment or small threshold $C$ induces considerable BA noise.

Three approaches have mainly been proposed to remove noise events from a raw event stream.

\noindent\textbf{Spatial low-pass filter:}\label{method_lowp}
The classical but still widely-used approach is spatial low-pass filtering.
In this approach, a raw event stream within a short time window is aggregated to obtain a 2-D image.
A low-pass spatial filter is applied to a set of 2-D images, and smoothed images corresponding to the denoised events are obtained \cite{xieImprovedApproachVisualizing2017, xieDVSImageNoise2018}.
However, it often ignores the temporal behavior of the events due to aggregation.
\medskip

\noindent\textbf{Density-based analysis:}\label{method_Den}
The second method is also a filtering-based method that uses the cumulative density of events in spatiotemporal space \cite{delbruckFramefreeDynamicDigital2008,fengEventDensityBased2020a,zhangNeuromorphicImagingDensitybased2023}.
The density is calculated by counting the number of neighboring events for each event within a small radius sphere. 
However, they typically remove the real events located at the border of the object or mistakenly detect hot pixel noise as real events.
\medskip

\noindent\textbf{Deep neural networks:}\label{method_deep}
As in image processing methods for frame-based cameras, deep neural networks are also used for event denoising.
For example, an event denoising method using a graph neural network (GNN) is proposed in \cite{alkendiNeuromorphicCameraDenoising2022}. 
However, it typically requires a large amount of annotated training data: This is not often the case for event-based cameras.

The above-mentioned existing methods typically focus on linear or nonlinear filtering in the local 3-D space.
However, analyzing spatiotemporal correlation with the local window results in a limited performance when noise events are spatiotemporally correlated with each other.
To address this issue, in the following section, we introduce a \textit{graph} to analyze the local and global correlations of events simultaneously.

\section{Spectral Graph Theory and Graph Construction Method}\label{sec:preliminary}
In this section, we first introduce the basics of spectral graph theory and the definition of the Fiedler vector \cite{ortegaGraphSignalProcessing2018,fiedlerAlgebraicConnectivityGraphs1973}.
We then review the $\epsilon$-neighbor graph construction method.

\subsection{Basics of Spectral Graph Theory}\label{sub:graph_basic}
Let $\mathcal{G}=(\mathcal{V},\mathcal{E},\mathbf{W})$ be a weighted undirected graph, where $\mathcal{V}$ and $\mathcal{E}$ are the sets of nodes and edges, and $\mathbf{W}$ is a weighted adjacency matrix. The $(i,j)$-element of $\mathbf{W}$ is positive if $(i,j)$ is connected, i.e., $[\mathbf{W}]_{i,j}>0$ if $(i,j)\in\mathcal{E}$, and $[\mathbf{W}]_{i,j}=0$ otherwise. The number of nodes is denoted as $|\mathcal{V}|=N$. The degree matrix is defined as $[\Dm]_{n,n} = \sum_{m}[\Wm]_{n,m}$. Then, a combinatorial graph Laplacian matrix is defined as $\mathbf{L}=\mathbf{D}-\mathbf{W}$, and symmetric normalized graph Laplacian matrix is defined as $\underline{\mathbf{L}}=\mathbf{D}^{-1/2}\mathbf{L}\mathbf{D}^{-1/2}$. 
We refer to $\mathbf{L}$ and $\underline{\mathbf{L}}$ as graph variation operators. 

Since $\mathbf{L}$ (and also $\underline{\mathbf{L}}$) is symmetric and positive semidefinite, they are diagonalizable: The eigendecomposition of $\mathbf{L}$ is given by $\mathbf{L}=\mathbf{U}\mathbf{\Lambda}\mathbf{U}^\top$, where $\mathbf{U}=[\mathbf{u}_{0},\ldots,\mathbf{u}_{N-1}]$ is a collection of eigenvectors and $\mathbf{\Lambda}=\text{diag}(\lambda_0,\ldots,\lambda_{N-1})$ is a diagonal eigenvalue matrix. Without loss of generality, we can assume that these eigenvalues are ordered as $0=\lambda_0\leq\lambda_{1}\leq\cdots\leq\lambda_{N-1}=\lambda_{\max}$.
A graph signal is defined as a mapping from a node to a real number, i.e., $x:\mathcal{V}\to\mathbb{R}$.
Simply speaking, the graph signal is represented by a vector $\mathbf{x} \in \mathbb{R}^N$ and $x_i$ is viewed as the signal value on the $i$th node.


The Fiedler vector of a graph is formally defined as follows:
\begin{definition}[Fiedler vector]\label{def:gfv}
Let $\mathbf{L}$ be the graph Laplacian matrix of an undirected connected graph $\mathcal{G}$. Then, the Fiedler vector is given by the eigenvector associated with the second smallest eigenvalue of $\mathbf{L}$ \cite{fiedlerAlgebraicConnectivityGraphs1973}, i.e., 
\begin{equation}\label{eq:bin_clustering_relax}
    \mathbf{u}_{\text{FV}}\coloneqq \argmin_{\mathbf{x},\|\mathbf{x}\|=1}\mathbf{x}^\top\mathbf{L}\mathbf{x}\quad \text{s.t. }\mathbf{x}^\top\mathbf{L}\mathbf{x}> 0.
\end{equation}
The second smallest eigenvalue is called \textit{algebraic connectivity} \cite{fiedlerAlgebraicConnectivityGraphs1973}.
\end{definition}


\subsection{$\epsilon$-Neighbor Graph Construction}\label{sub:graph_const}
In some applications, the graph may not be given a priori, such as point clouds and sensor networks.
Event streams are also the case.
Therefore, we need to construct an appropriate graph from the event stream based on some requirements (e.g., the metric of distance and degree of nodes).
In this paper, we use the $\epsilon$-neighbor graph construction.

$\epsilon$NG is one of the popular graph construction methods.
It connects nodes to their neighbors within the radius $\epsilon$.
Suppose that we have a collection of features $\Xm = [\xv_1, \dots, \xv_M]\in\mathbb{R}^{N\times M}$, where $N$ is the dimensionality of each feature and $M$ is the number of the features.
Then, an off-diagonal element in a binary adjacency matrix of the $\epsilon$NG is expressed as follows:
\begin{equation}
    [\mathbf{A}_{\epsilon\text{NG}}]_{i,j\neq i}=
    \begin{cases}
        1 & \text{if }\|[\mathbf{X}]_{i,:}-[\mathbf{X}]_{j,:}\|\leq \epsilon\\
        0 & \text{otherwise}.
    \end{cases}\label{eq:epsilon-graph}
\end{equation}
The value of $\epsilon$ controls the connectivity of the $\epsilon$NG.



\section{Event Denoising Using Graph Spectral Features}\label{sec:proposed}
In this section, we propose a denoising method for event-based cameras using graph spectral features. The overview of the proposed method is illustrated in Fig.~\ref{fig:over_flow}.

The proposed method mainly consists of three steps.


\begin{enumerate} [resume,listparindent=1em]
\item \emph{Graph construction}: 
An $\epsilon$NG is constructed from observed noisy events. We denote the corresponding graph Laplacian as $\mathbf{L}_{\epsilon\text{NG}}$.
\item \emph{Calculation of the eigenvectors}: 
The eigenvectors of $\mathbf{L}_{\epsilon\text{NG}}$ is calculated. 
\item \emph{Large components detection}: 
We obtain denoised events by extracting indices of nonzero elements in eigenvectors for small eigenvalues as large connected component detection in the $\epsilon$NG. 
\end{enumerate}


\noindent In the following, we describe each building block in detail.

\begin{figure}[t]
\centering
\includegraphics[width = \linewidth]{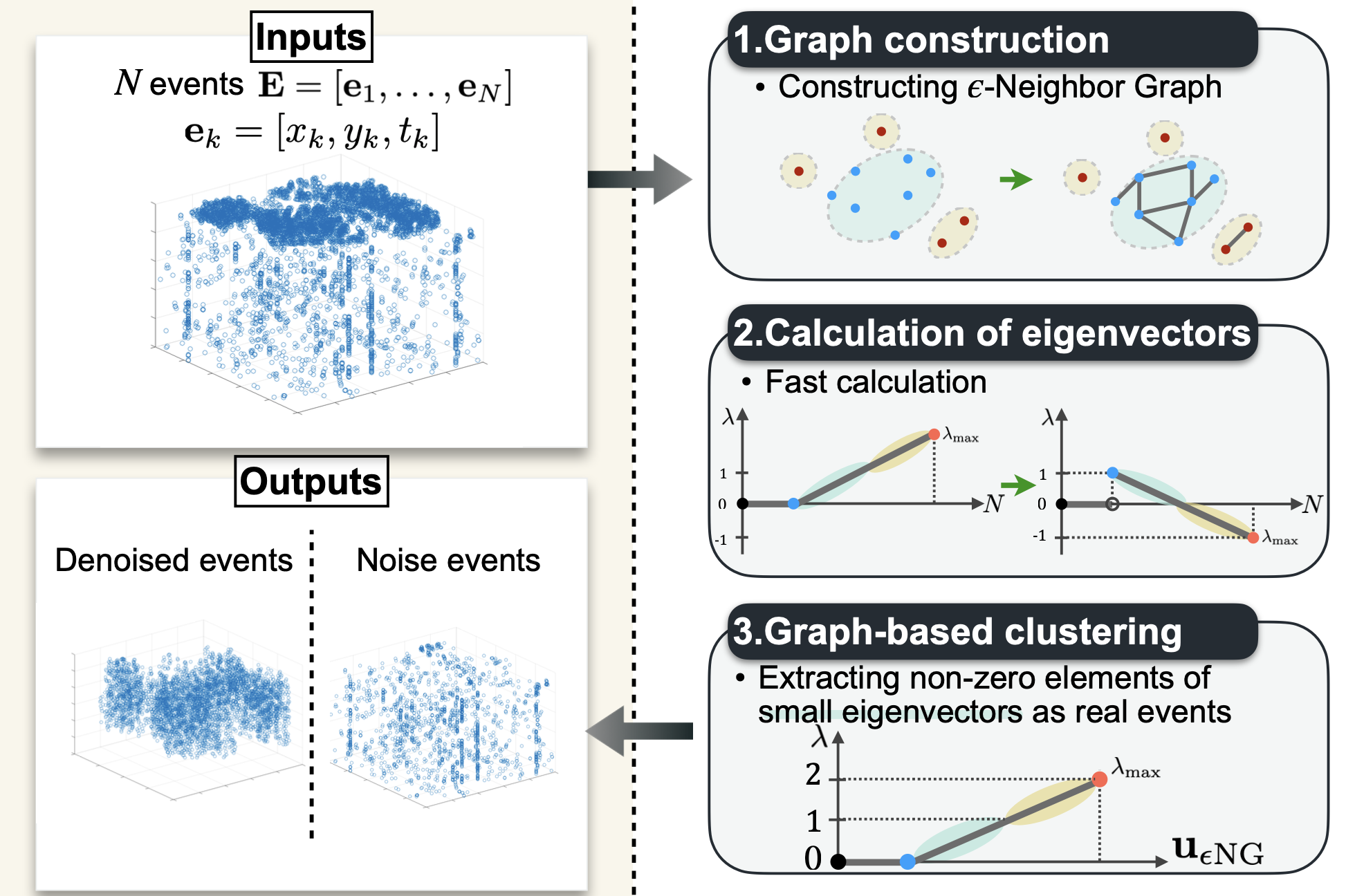} 
\caption{Overview of the proposed method.}
\label{fig:over_flow}
\end{figure}

\subsection{Graph Construction}\label{sub:estimate_e}

In the graph construction step, we use $\epsilon$NG introduced in Section \ref{sub:graph_const}. Although our focus is on $\epsilon$NG, we compare this approach with other typical graph construction methods in Appendix.

The construction of an $\epsilon$NG consists of two steps: 1) temporal scale normalization and 2) estimation of $\epsilon$.


\subsubsection{Temporal scale normalization} Similar to \cite{biGraphbasedObjectClassification2019,schaeferAegnnAsynchronousEventbased2022}, the temporal coordinates in $\mathbf{E}$ are scaled by a factor $\beta$ to ensure that the temporal and spatial coordinates are within a comparable range. 
This scaling can be written as
$\widetilde{\mathbf{E}} = \mathrm{diag}(1,1,\beta)\mathbf{E}$ where the scaling factor $\beta>0$ is empirically set.

\subsubsection{Estimation of $\epsilon$} 
Recall that real events typically form large clusters in the 3-D space, while noise events scatter in fragmented pieces.
To encode the above characteristic into the graph structure,
we utilize the spatiotemporal density. 
Inspired by \cite{ester1996density}, we measure the spatiotemporal density of events by averaging distances between events in their $k$ nearest neighbors.
The density of the $n$th event is defined as follows.
\begin{equation}\label{eq:local_density}
    d_n \coloneqq \frac{1}{k}\sum_{m\in\mathcal{N}_n} \|[\widetilde{\mathbf{E}}]_{:,n}-[\widetilde{\mathbf{E}}]_{:,m}\|^2,
\end{equation}
where $\mathcal{N}_n$ is the $k$ nearest events in the 3-D space for the $n$th node, i.e., $|\mathcal{N}_n|=k$ for all $n$.
That is, an event with a large $d_n$ probably corresponds to an isolated event while a small $d_n$ indicates an event with dense neighbors. 

We then estimate the best $\epsilon$ from $d_n$.
For brevity, suppose that $d_n$ in \eqref{eq:local_density} is sorted in nonincreasing order, i.e., $d_1\geq d_2\geq \dots\geq d_N$. Then, we determine the best $\epsilon$ as the \textit{knee point} \cite{ester1996density} of the sequence of $d_n$:
    \begin{equation}\label{eq:epsilon*}
        \epsilon^* = \argmax_{d_n} L(n)-d_n,
    \end{equation}
    where $L(n)=\frac{d_N-d_1}{N-1}(n-1)+d_1$ is the straight line connecting $(1,d_1)$ and $(N,d_N)$.
Note that, as shown in Fig.~\ref{fig:knee_point}, the knee point is one of $d_n$'s that is maximally away from $L(n)$.

Densely clustered events produce a large number of similar, small $d_n$ values. These $d_n$s form a plateau in the curve by connecting them. In contrast, sparse events result in a steep descent. Thus, the knee point is expected to represent the transition in $d_n$ that separates densely clustered events from sparsely distributed ones (see Fig.~\ref{fig:knee_point}). By setting $\epsilon=\epsilon^*$, this difference of densities is encoded in the constructed graph: Events with dense neighbors are likely to form large connected components, whereas sparse events form small connected components or even remain isolated in $\epsilon$NG.

Finally, we construct the $\epsilon$NG with $\epsilon^*$ in \eqref{eq:epsilon*} (see Section \ref{sub:graph_const}). We assign the edge weights based on the radial basis function (RBF) kernel, i.e., $[\Wm]_{i,j}=\exp(-\gamma\|\widetilde{\mathbf{E}}_{:,i}-\widetilde{\mathbf{E}}_{:,j}\|^2)$ where $\gamma>0$ is the parameter controlling the decay of an RBF kernel. Hereafter, we will use $\mathbf{L}$ as the graph operator of the $\epsilon$NG.

\begin{figure}[t]
\centering
\includegraphics[width = 0.7\linewidth]{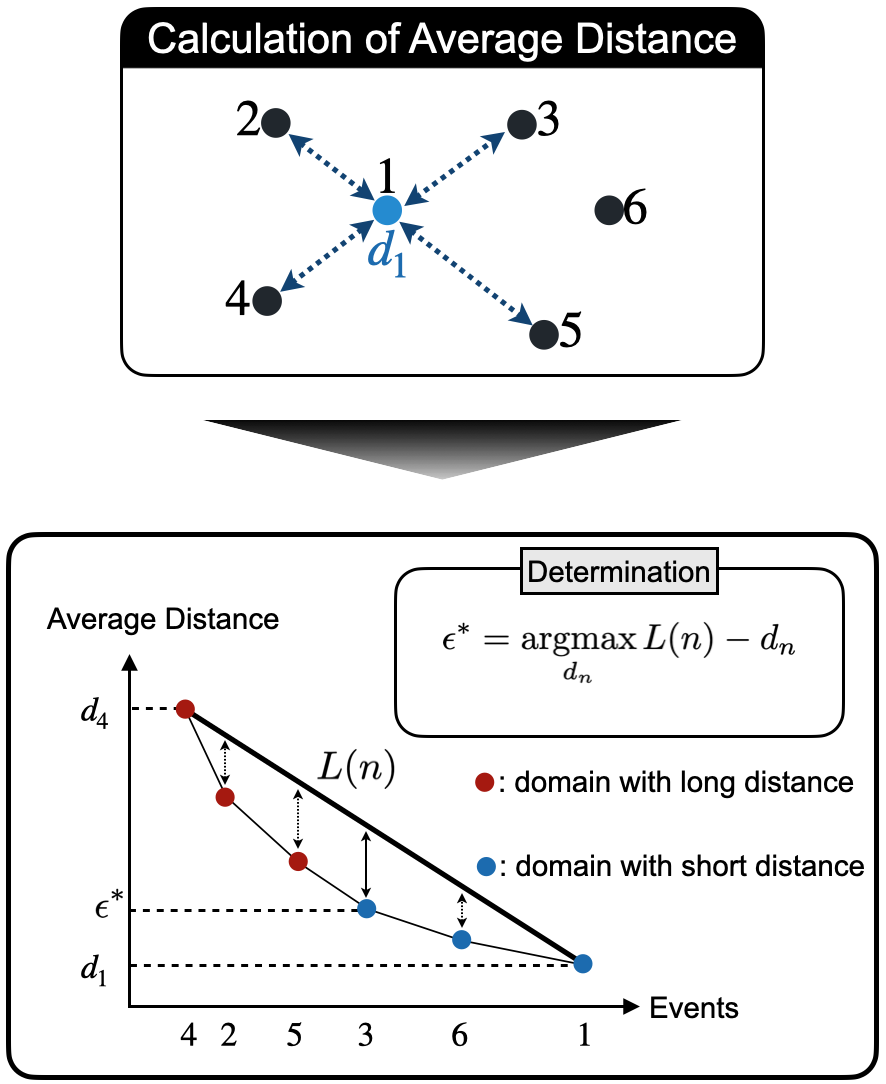} 
\caption{Overview of the estimation of $\epsilon^*$.}
\label{fig:knee_point}
\end{figure}

\subsection{Calculation of Eigenvectors}\label{sub:power}

In the second step, we calculate eigenvectors of the constructed $\epsilon$NG.
Specifically, we use (approximated) eigenvectors corresponding to small eigenvalues.
Two approaches exist for the eigenvector computation.
A naive eigendecomposition is highly accurate but demands a high computational cost. In contrast, approximate iterative methods are computationally efficient with some approximation error. 


Note that, approximation methods like the power method are usually designed to compute eigenvectors associated with large eigenvalues.
Therefore, we need to customize the approximation methods for our purpose.

In a straightforward approach to approximate eigenvectors associated with small eigenvalues, the classical power method is applied to $\mathbf{I}-\underline{\mathbf{L}}=\mathbf{I}-\mathbf{D}^{-1/2}\mathbf{L}\mathbf{D}^{-1/2}$. However, it assumes that $\underline{\mathbf{L}}=\mathbf{D}^{-1/2}\mathbf{L}\mathbf{D}^{-1/2}$ only has a \textit{single} zero eigenvalue, i.e., the graph must be connected.
Note that the $\epsilon$NG we used here is not always connected: The extension of the power method in this case is not trivial.

To overcome the limitations above, we propose a customized power method for the efficient calculation of $\mathbf{u}_{\text{FV}}$ on disconnected graphs:
\begin{proposition}[Customized power method]\label{prop:power_fiedler}
Let $\mathbf{L}$ be the graph Laplacian of an undirected weighted graph $\mathcal{G}$. Consider the following matrix:
\begin{equation}\label{eq:modified_eig}
    \mathbf{S}=\left(\mathbf{I}-\frac{2}{\rho_{\max}}\mathbf{L}\right)\left(\mathbf{I}-\left(\mathbf{I}-\frac{2}{\rho_{\max}}\mathbf{L}\right)^{\omega}\right),
\end{equation}
where $\rho_{\max}$ is the upper bound of $\lambda_{\max}(\mathbf{L})$ defined in \cite{giraultStationaryGraphSignals2015} and $\omega>1$ is some constant. If $\omega$ is sufficiently large, the following update converges to the Fielder vector $\mathbf{u}_{\text{FV}}$.
\begin{equation}\label{eq:power_method}
\mathbf{x}^{l+1}=\mathbf{S}\mathbf{x}^l/\|\mathbf{S}\mathbf{x}^l\|
\end{equation}
where $l$ is the iteration number.
\end{proposition}

\begin{proof}
First, we consider the following problem.
\begin{equation}\label{eq:standard_evd}
    \max_{\mathbf{x},\|\mathbf{x}\|=1}\mathbf{x}^\top\mathbf{S}\mathbf{x}.
\end{equation}
Since \eqref{eq:standard_evd} is convex, it can be solved by a gradient descent method \cite{journeeGeneralizedPowerMethod}, i.e., 
\begin{subequations}
\begin{align}
    \mathbf{z}^{l+1}=&\mathbf{x}^{l}-\delta(\mathbf{x}^l-\mathbf{S}\mathbf{x}^l)\label{eq:evd_grad_1}\\
    \mathbf{x}^{l+1}=&\mathbf{z}^{l+1}/\|\mathbf{z}^{l+1}\|,\label{eq:evd_grad_2}
\end{align}
where $\delta$ is the appropriate step size of the update. By setting $\delta=1$, \eqref{eq:evd_grad_1} and \eqref{eq:evd_grad_2} are reduced to \eqref{eq:power_method}.
\end{subequations}
Note that the solution of \eqref{eq:standard_evd} is nothing but the eigenvector associated with the maximum eigenvalue of $\mathbf{S}$.
Since $\mathbf{S}$ is diagonalizable by $\mathbf{U}$ (eigenvector matrix of $\mathbf{L}$),
the diagonal eigenvalue matrix of $\mathbf{S}$ is given by
\begin{equation}\label{eq:eigen_smat}
    \bm{\Lambda}_S= \left(\mathbf{I}-\frac{2}{\rho_{\max}}\mathbf{\Lambda}\right)\left(\mathbf{I}-\left(\mathbf{I}-\frac{2}{\rho_{\max}}\bm{\Lambda}\right)^\omega\right).
\end{equation}
The equation \eqref{eq:eigen_smat} transforms the smallest non-zero eigenvalue in $\mathbf{\Lambda}$ into the largest value in $\bm{\Lambda}_S$.
Note that $|1-\frac{2}{\rho_{\max}}\lambda_i| < 1$ is always satisfied.
By setting $\omega\rightarrow\infty$, \eqref{eq:eigen_smat} converges to
\begin{equation}\label{eq:eigen_smat_2}
    \lim_{\omega\rightarrow \infty}\bm{\Lambda}_S= \left(\mathbf{I}-\frac{2}{\rho_{\max}}\mathbf{\Lambda}\right)H(\mathbf{\Lambda}),
\end{equation}
where $H(\mathbf{\Lambda})$ is a graph high-pass filter such that 
\begin{equation}
        H(\lambda_i)=
        \begin{cases}
            0 & \text{if } \lambda_i=0\\
            1 & \text{otherwise}.
        \end{cases}
    \end{equation}
As a result, $(1-\frac{2}{\rho_{\max}}\lambda_{\max})H(\lambda_{\max})$ in \eqref{eq:eigen_smat_2} corresponds to the minimum non-zero eigenvalue of $\mathbf{L}$, which is associated with the Fiedler vector. This completes the proof.
\end{proof}

Proposition \ref{prop:power_fiedler} shows that the Fiedler vector can be efficiently calculated with the power method by customizing the graph Laplacian so that the customized matrix has the non-zero eigenvalues with the reversed order of the original one.
Since \eqref{eq:standard_evd} is convex, the initial value $\mathbf{x}^0$ in \eqref{eq:power_method} can be set arbitrarily.
The value of $\omega$ corresponds to a trade-off between computational cost and the accuracy of $\uv_{\text{FV}}$.
In this paper, we empirically set $\omega$.

If multiple eigenvectors of $\underline{\mathbf{L}}$ are required, one option would be to combine the customized power method with the deflation technique \cite{burrage1998deflation}. The deflation is performed with the following update:
\begin{equation} \label{eq:deflation}
    \mathbf{S}^{(i+1)}=\mathbf{S}^{(i)}-\lambda^{(i)}\mathbf{x}^{(i)}\mathbf{x}^{(i)\top}
\end{equation}
where $\lambda^{(i)}$ is the $i$th largest eigenvalue and $\mathbf{x}^{(i)}$ is its corresponding eigenvector. The initial operator $\mathbf{S}^{(0)}$ is set to $\mathbf{S}$ derived from \eqref{eq:modified_eig}.
This update orthogonalizes $\mathbf{S}^{(i+1)}$ from the previously extracted eigenvector $\mathbf{x}^{(i)}$.
By applying \eqref{eq:deflation} repeatedly, multiple eigenvectors can be computed.

Another option would be applying existing eigensolver algorithms, such as the locally optimal block preconditioned conjugate gradient (LOBPCG) method \cite{knyazevOptimalPreconditionedEigensolver2000}, to the transformed graph operator $\mathbf{S}$ derived in \eqref{eq:modified_eig}.
We compare the performance according to the eigenvector calculation methods in Section \ref{sec:experiment}.

Generally, the magnitude of eigenvectors associated with small eigenvalues encodes the strength of connectivity \cite{fiedlerAlgebraicConnectivityGraphs1973}. Particularly, when the graph is constructed based on the spatiotemporal densities of events, the elements in these eigenvectors measure their inter-cluster densities. In the following section, we reveal this characteristic and detect real events from the obtained eigenvectors.

\subsection{Event Denoising as Large Component Detection}\label{sec:component_detection}
We detect real events using eigenvectors associated with small eigenvalues. We first consider the case in which real events form a single large connected component within the $\epsilon$NG. We then extend our focus to the detection of real events forming multiple disjoint clusters.

\subsubsection{Detection of Single Large Component}\label{sub:detection}
Real events may form a large graph (real graph) on the $\epsilon$NG while noise events could be separated as many isolated graphs (noise graphs).
This leads to that, each noise graph contains significantly fewer nodes compared to the real graph.
Here, we detect the real graph as the largest connected component. To this aim, we clarify the behavior of the largest connected components in the Fiedler vector $\mathbf{u}_{\text{FV}}$. 

For simplicity, we can suppose that two undirected subgraphs, $\mathcal{G}_{r}$ with $\alpha N$ nodes and $\mathcal{G}_{n}$ with $(1-\alpha)N$ nodes, are in $\mathcal{G}$ while $\mathcal{G}_{r}$ and $\mathcal{G}_{n}$ are disconnected.
First, we assume the following characteristic regarding the cardinalities of events:
\begin{assump}\label{assum:G}
The parameter $\alpha$ is close to $1$.
\end{assump}

\noindent Assumption \ref{assum:G} corresponds to that $\mathcal{G}_{r}$ is a real graph and $\mathcal{G}_{n}$ is a noise graph.
Let us denote the diameter of a graph $\mathcal{G}$ by $D_{\mathcal{G}}=\max_{(u,v)\in\mathcal{E}}d(u,v)$, where $d(u,v)$ is the sum of edge weights along the shortest path between $u$ and $v$.
Following from Assumption \ref{assum:G}, $D_{\mathcal{G}_{n}}$ should be significantly smaller than that of the real graph.
In addition, noise events could be distributed in two ways: as isolated nodes or as very small clusters composed of only a few nodes\footnote{Hot pixel noise may form independent small clusters since it tends to distribute more sparsely than real events along the temporal direction.}. In both cases, we can assume the following property of $\mathcal{G}_{n}$:
\begin{assump}\label{assum:noise_G}
The noise graph $\mathcal{G}_{n}$ is a complete graph.
\end{assump}



With Assumptions \ref{assum:G} and \ref{assum:noise_G}, the magnitudes of the Fiedler vector for each element can be determined by the following proposition.
\begin{proposition}\label{prop:lambda_comparison}
    Suppose that we have the whole graph Laplacian $\mathbf{L}=\emph{\text{diag}}(\mathbf{L}_{r},\mathbf{L}_{n})$ where $\mathbf{L}_{r}$ and $\mathbf{L}_{n}$ are graph Laplacians of $\mathcal{G}_{r}$ and $\mathcal{G}_{n}$, respectively.
    Let us also denote Fiedler vectors associated with the subgraphs, $\mathcal{G}_{r}$ and $\mathcal{G}_{n}$, by $\mathbf{u}_{\text{FV}_r}$ and $\mathbf{u}_{\text{FV}_n}$, respectively.
    The solution of \eqref{eq:bin_clustering_relax} for $\mathbf{L}$, i.e., the Fiedler vector of the whole graph Laplacian $\mathbf{L}$, is given by $[\mathbf{u}_{\text{FV}_r}^\top,\mathbf{0}^\top]^\top$ if Assumption \ref{assum:G} and \ref{assum:noise_G} hold.
    
\end{proposition}


\begin{proof}
    For conciseness, we consider symmetric normalized Laplacian of $\mathcal{G}_{r}$ and $\mathcal{G}_{n}$, i.e., $\underline{\mathbf{L}_{r}}$ and $\underline{\mathbf{L}_{n}}$\footnote{This relies on the fact that the order of magnitudes of the eigenvalues do not change between a combinational graph Laplacian and a symmetric normalized Laplacian.}.
    Denoting the algebraic connectivity of $\underline{\mathbf{L}_{r}}$ and $\underline{\mathbf{L}_{n}}$ by $a(\mathcal{G}_{r})$ and $a(\mathcal{G}_{n})$, respectively,
    candidates for $\mathbf{u}_{\text{FV}}$
    are given as follows:
    \begin{equation}\label{eq:fiedler_single}
        \mathbf{u}_{\text{FV}}=
        \begin{cases}
            [\mathbf{u}_{\text{FV}_r}^\top,\mathbf{0}^\top]^\top\quad & \text{if } a(\mathcal{G}_{r}) < a(\mathcal{G}_{n})\\
            [\mathbf{0}^\top,\mathbf{u}_{\text{FV}_n}^\top]^\top\quad & \text{otherwise}.
        \end{cases}
    \end{equation}
    Under Assumption \ref{assum:noise_G}, the eigenvalues of $\underline{\mathbf{L}_{n}}$ are $0$ and $(1-\alpha)N/((1-\alpha)N-1)$ with multiplicity $(1-\alpha)N-1$ \cite{chungSpectralGraphTheory1997}.
    Accordingly, $a(\mathcal{G}_{n})$ satisfies the following inequality.
    \begin{equation}\label{eq:lower_b}
        a(\mathcal{G}_{n}) = \frac{(1-\alpha)N}{(1-\alpha)N-1} > 1.
    \end{equation}

    We then consider $\mathcal{G}_{r}$.
    Let $D_{\max}$ and $D_{\mathcal{G}_{r}}$ be the maximal degree and the diameter of $\mathcal{G}_{r}$, respectively. 
    The upper bound of $a(\mathcal{G}_{r})$ can be given as follows \cite{chungSpectralGraphTheory1997}:
    \begin{equation}\label{eq:upper_b}
            a(\mathcal{G}_{r}) \leq 1-2\frac{\sqrt{D_{\max}-1}}{D_{\max}}\left(1-\frac{2}{D_{\mathcal{G}_{r}}}\right)+\frac{2}{D_{\mathcal{G}_{r}}}.
    \end{equation}
    From Assumption \ref{assum:G}, we can assume that $D_{\mathcal{G}_{r}}\gg 1$ and this leads to
    $2/D_{\mathcal{G}_{r}}\approx0$.
    Therefore, \eqref{eq:upper_b} can be approximated as follows:
    \begin{equation}\label{eq:upper_b_relax}
        a(\mathcal{G}_{r}) \leq 1- 2\frac{\sqrt {D_{\max}-1}}{D_{\max}}<1.
    \end{equation}
    From \eqref{eq:lower_b} and \eqref{eq:upper_b_relax}, we can obtain the following relationship between $a(\mathcal{G}_{n})$ and $a(\mathcal{G}_{r})$:
    \begin{equation}
        0 <a(\mathcal{G}_{r}) <1<a(\mathcal{G}_{n}).
    \end{equation}
    This completes the proof.
\end{proof}

In \eqref{eq:fiedler_single}, real events correspond to nonzero entries in $\mathbf{u}_{\text{FV}}$. 
Proposition \ref{prop:lambda_comparison} shows that we can detect real events as the largest component in the raw event stream with the magnitude of the elements in $\mathbf{u}_{\text{FV}}$. 


Let $\mathbf{y}\in\{0,1\}^N$ be the binary labels of events, where $[\mathbf{y}]_i=1$ and $[\mathbf{y}]_i=0$ indicate real and noise events, respectively. Consequently, we detect the large component of the graph as follows:
\begin{equation}\label{eq:large_comp_detect}
    [\mathbf{y}]_i=
    \begin{cases}
       1 & \text{if }|[\mathbf{u}_{\text{FV}}]_i| \geq \eta\\
       0 & \text{otherwise},
    \end{cases}
\end{equation}
where $\eta>0$ is a small constant. 


\subsubsection{Detection of Multiple Large Connected Components}\label{sub:multi-denoising}
Until now, we have assumed that real events form a single large component in the constructed graph.
We extend Proposition \ref{prop:lambda_comparison} to the detection of multiple real graphs. 

We consider $n_r$ connected components $\{\mathcal{G}_{r_i}\}_{i=1}^{n_r}$ in $\mathcal{G}$, each containing $\alpha_i N$ nodes, and $n_n$ connected components $\{\mathcal{G}_{n_j}\}_{j=1}^{n_n}$, each containing $\beta_j N$ nodes. We assume that $\{\mathcal{G}_{r_i}\}_{i=1}^{n_r}$ and $\{\mathcal{G}_{n_j}\}_{j=1}^{n_n}$ are disconnected with each other, which satisfy $\sum_{i=1}^{n_r}\alpha_{i}+\sum_{j=1}^{n_n}\beta_{j}=1$. Similar to Assumption \ref{assum:G}, we assume the following characteristic regarding the cardinalities of evens.

\begin{assump}\label{assum:G_m}
 $\alpha_i \gg \beta_j $ for all $i$ and $j$.
\end{assump}

\noindent Under Assumption \ref{assum:G_m}, we assume that $\{\mathcal{G}_{r_i}\}_{i=1}^{n_r}$ and $\{\mathcal{G}_{n_j}\}_{j=1}^{n_n}$ correspond to real graphs and noise graphs, respectively. 
Therefore, Assumption \ref{assum:G_m} requires that all real graphs are much larger than noise graphs. 
For detecting multiple real graphs, we extend Proposition \ref{prop:lambda_comparison} to the following corollary:
\begin{corollary}\label{coro:extended}
     Let us denote the normalized graph Laplacians of real graphs $\{\mathcal{G}_{r_i}\}_{i=1}^{n_r}$ and noise graphs $\{\mathcal{G}_{n_j}\}_{j=1}^{n_n}$ by $\{\underline{\mathbf{L}}_{r_i}\}_{i=1}^{n_r}$ and $\{\underline{\mathbf{L}}_{n_j}\}_{j=1}^{n_n}$, respectively. The normalized graph Laplacian of $\mathcal{G}$ is given by $\underline{\mathbf{L}}=\emph{\text{diag}}(\underline{\mathbf{L}}_{r_1},\cdots,\underline{\mathbf{L}}_{n_r},\underline{\mathbf{L}}_{n_1},\cdots,\underline{\mathbf{L}}_{n_n})$. 
    If Assumptions \ref{assum:noise_G} and \ref{assum:G_m} hold, the eigenvectors corresponding to the eigenvalues within interval $(0, 1)$ are eigenvectors associated with $\{\mathcal{G}_{r_i}\}_{i=1}^{n_r}$.
\end{corollary}

\begin{proof}
    We denote the nodes in $\mathcal{G}_{r_i}$ and $\mathcal{G}_{n_j}$ by $\mathcal{V}_{r_i}$ and $\mathcal{V}_{n_j}$, respectively.
    Since $\mathcal{G}_{r_i}$ and $\mathcal{G}_{n_j}$ do not have connections between them, an eigenvector of $\underline{\mathbf{L}}$, associated with $\mathcal{G}$, has the form of
    \begin{equation}
        \underline{\mathbf{u}}=
        \mathbbm{1}_{\mathcal{C}}(\underline{\mathbf{u}}) \quad \exists\mathcal{C}\in \{\mathcal{V}_{r_i}\}_{i=1}^{n_r}\cup \{\mathcal{V}_{n_j}\}_{j=1}^{n_n},
    \end{equation}
    where $\mathbbm{1}_\mathcal{C}$ is the indicator function over the subset $\mathcal{C}$, i.e., $[\mathbbm{1}_\mathcal{C}(\mathbf{x})]_k=[\mathbf{x}]_k$ if $k\in\mathcal{C}$ and $[\mathbbm{1}_\mathcal{C}(\mathbf{x})]_k=0$ otherwise. Therefore, $[\underline{\mathbf{u}}]_{i\in\mathcal{C}}$ is identical to an eigenvector corresponding to the subgraph $\mathcal{C}$.
    
    From Assumptions \ref{assum:noise_G} and \ref{assum:G_m}, noise events form small complete graphs. Similar to \eqref{eq:lower_b}, we have $a(\mathcal{G}_{n_j}) > 1$ for all $n_j$.
    We also have $a(\mathcal{G}_{r_i}) < 1$ for all $r_j$ from \eqref{eq:upper_b_relax}.
    Consequently, we can conclude that
    \begin{equation}\label{eq:real_fv_int}
        0 < a(\mathcal{G}_{r_i}) < 1 < a(\mathcal{G}_{n_j})\hspace{1ex} \forall i\in\{1,\ldots,n_r\},\forall j\in\{1,\ldots,n_n\}.
    \end{equation}
\end{proof}

Followed by Corollary \ref{coro:extended}, we can identify real events by examining the non-zero elements of eigenvectors associated with eigenvalues within $(0, 1)$. Note that we do not require any knowledge of the number of objects. Then, we can reformulate \eqref{eq:large_comp_detect} as
\begin{equation}\label{eq:large_comp_detect_multi}
    [\mathbf{y}]_i=
    \begin{cases}
       1 & \text{if there exists } j \text{ satisfying }|[\underline{\mathbf{u}}_{j}]_i| \geq \delta \text{ and } \underline{\lambda}_j < 1\\
       0 & \text{otherwise},
    \end{cases}
\end{equation}
where $(\underline{\lambda}_j,\underline{\mathbf{u}}_j)$ is the $j$th eigenpair of the normalized graph Laplacian $\underline{\mathbf{L}}$. 

\section{Computational Complexity}\label{sec:complexity}
Here, we verify that the proposed customized power method is more computationally efficient than existing eigendecomposition (abbreviated as EVD) approaches by deriving the complexity of \eqref{eq:power_method}.

\noindent\textbf{Single connected component detection:} In a native computation based on EVD, obtaining eigenvectors of $\mathbf{L}_{\epsilon\text{NG}}$
    requires the complexity $O(N^3)$. The power method in \eqref{eq:power_method} requires the complexity $O(I(N^2+N))$, where $I$ is the number of iterations. In Proposition~\ref{prop:power_fiedler}, we reorder the eigenvectors of $\mathbf{L}_{\epsilon\text{NG}}$ to the reverse order with \eqref{eq:modified_eig}. This is implemented by the complexity $O(\omega N |\mathcal{E}_{\epsilon\text{NG}}|)$, where $|\mathcal{E}_{\epsilon\text{NG}}|$ is the number of edges in the $\epsilon$NG. As a result, the proposed method requires the complexity $O (N(IN+I+\omega |\mathcal{E}_{\epsilon\text{NG}}|))$. As long as the graph is sparse, it is computationally more efficient than the naive implementation of the EVD.

    \medskip

\noindent\textbf{Multiple connected component detection:}
    Suppose that $k$ eigenvectors are required to detect multiple objects.
    In practice, $k=20$ would be enough to detect real events from raw event streams.
    The deflation step in \eqref{eq:deflation} is implemented by the complexity $O(N^2)$, while the power method requires $O(I(N^2+N))$ to compute each eigenvector as mentioned above.
    Consequently, obtaining $k$ eigenvectors via a repeated application of the power method results in a total complexity of $O(kN(IN+N+I))$. Alternatively, LOBPCG requires computing multiple eigenvectors with the complexity of $O(kIN(N+k))$.
    The proposed method is computationally more efficient than the repeated power method if $kI<N+I$.

\section{Event Denoising Experiments}\label{sec:experiment}
In this section, we evaluate the effectiveness of the proposed method via denoising experiments with synthetic and real-world data.

\subsection{Denoising with Single Cluster of Real Events}

Here, we focus on the scenario in which the real events form a single cluster.

\begin{figure*}[t]
\centering
\subfigure[Ground truth]{\includegraphics[width = 0.48\columnwidth]{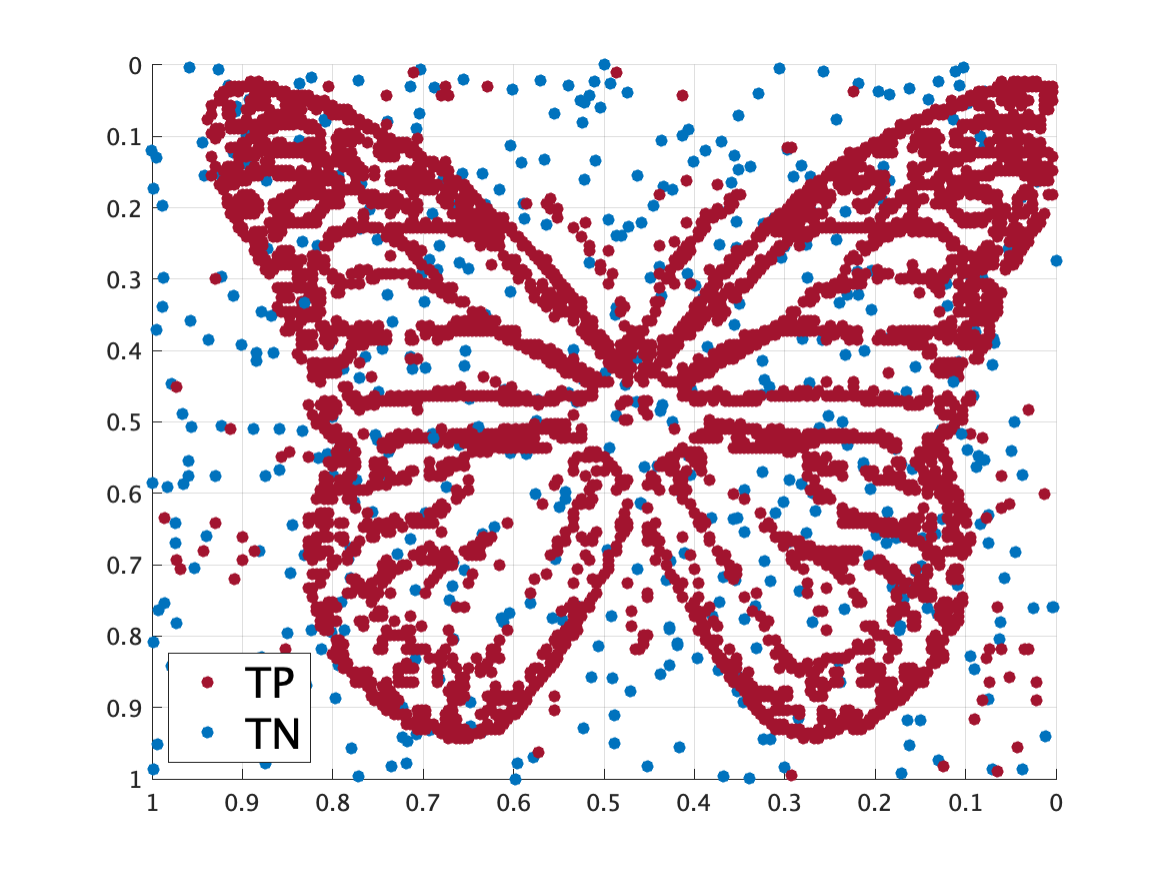} \label{fig:groundtruth_bike}} 
\subfigure[Denoised: proposed w/ EVD]{\includegraphics[width = 0.48\columnwidth]{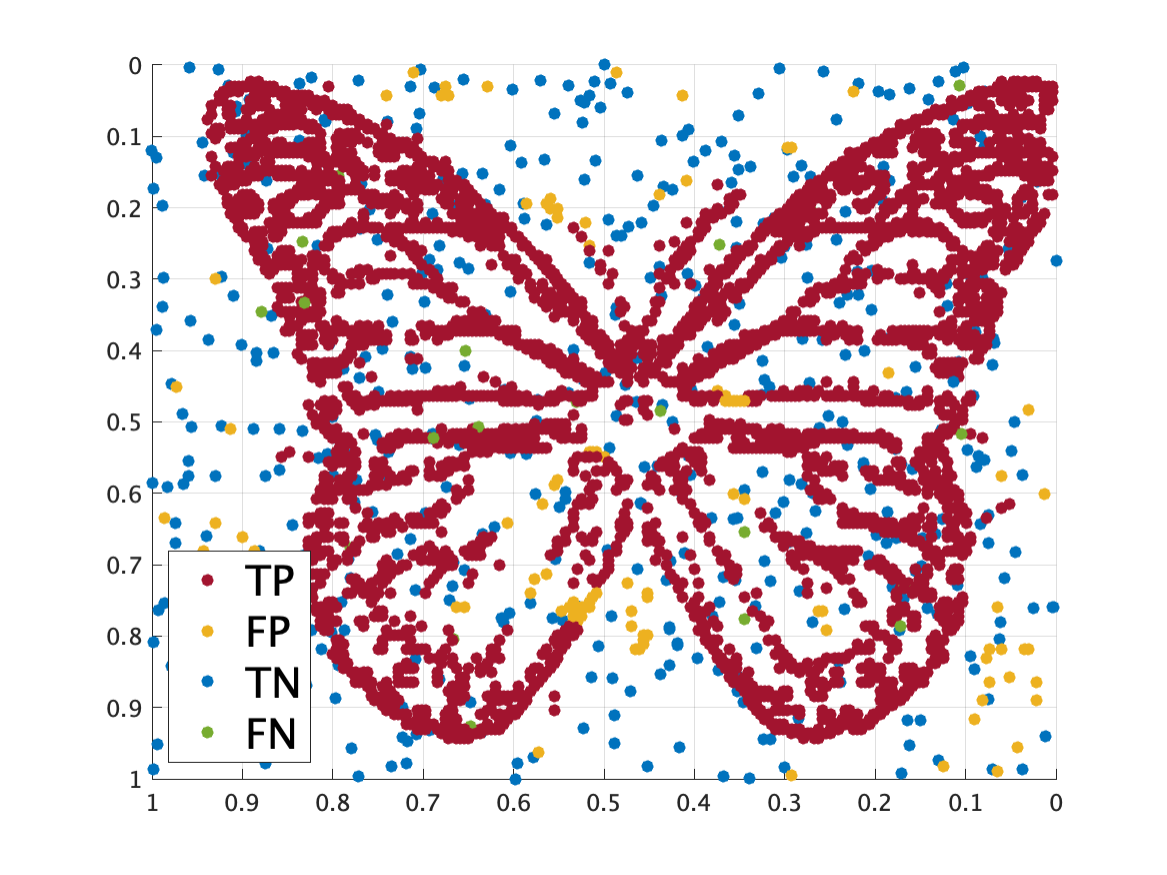} \label{fig:epsilon_bike}}
\subfigure[Denoised: proposed w/ Power Method]{\includegraphics[width = 0.48\columnwidth]{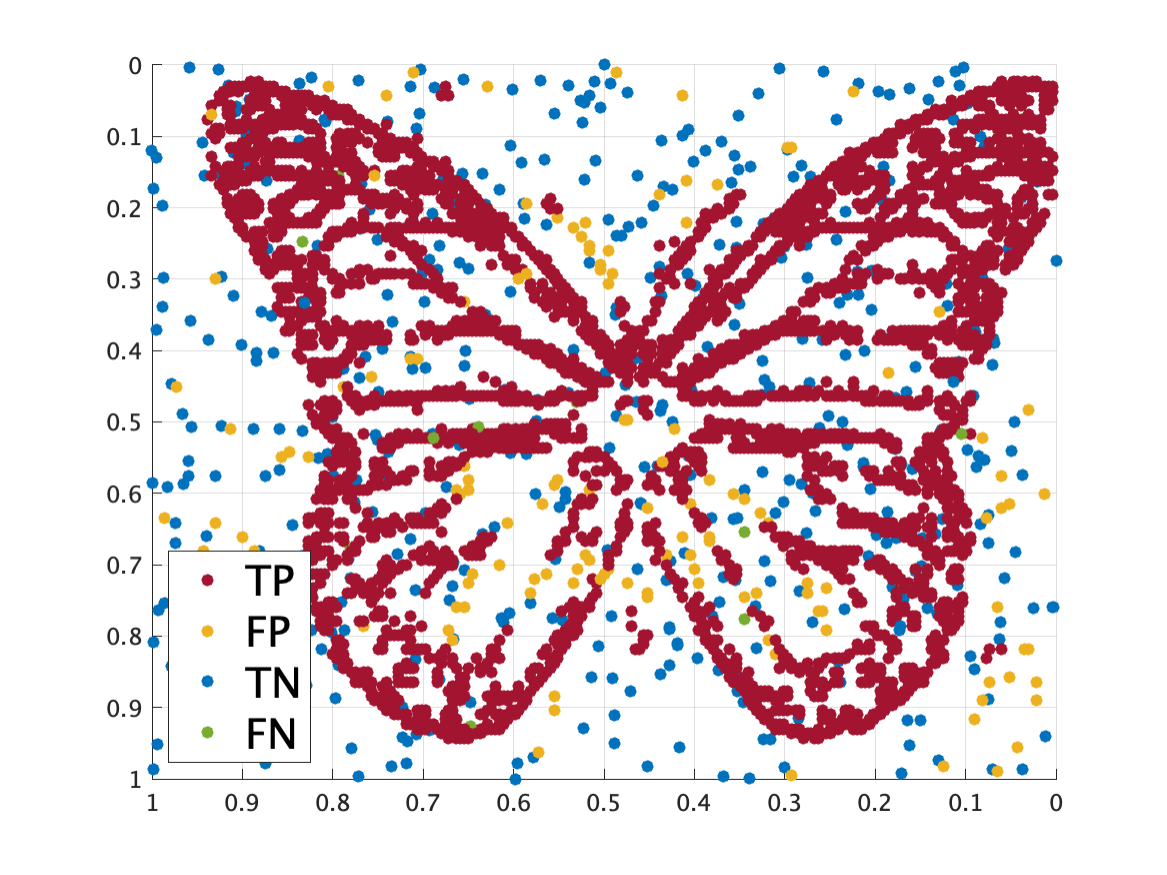} \label{fig:power_bike}}
\subfigure[Denoised: proposed w/ LOBPCG]{\includegraphics[width = 0.48\columnwidth]{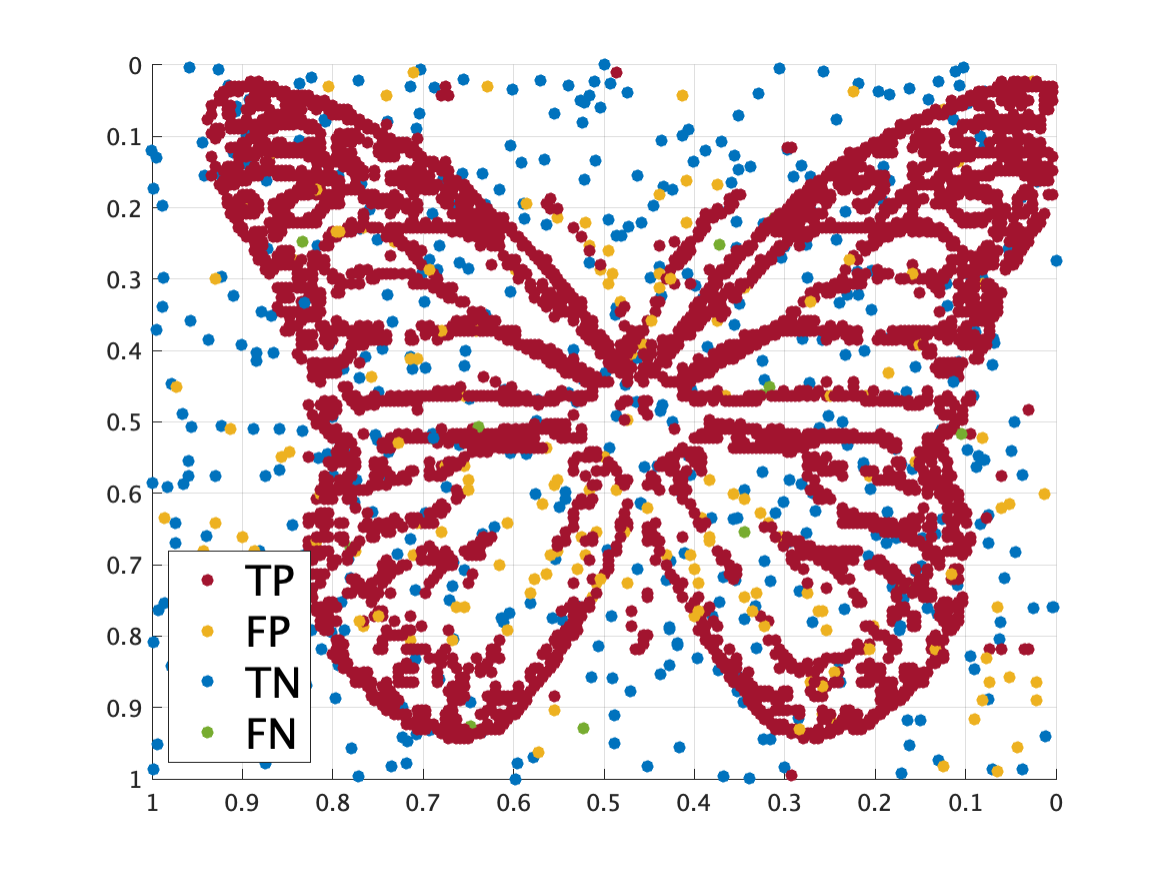} \label{fig:sin_lopcg}}
\subfigure[Denoised: Centroid Filter]{\includegraphics[width = 0.48\columnwidth]{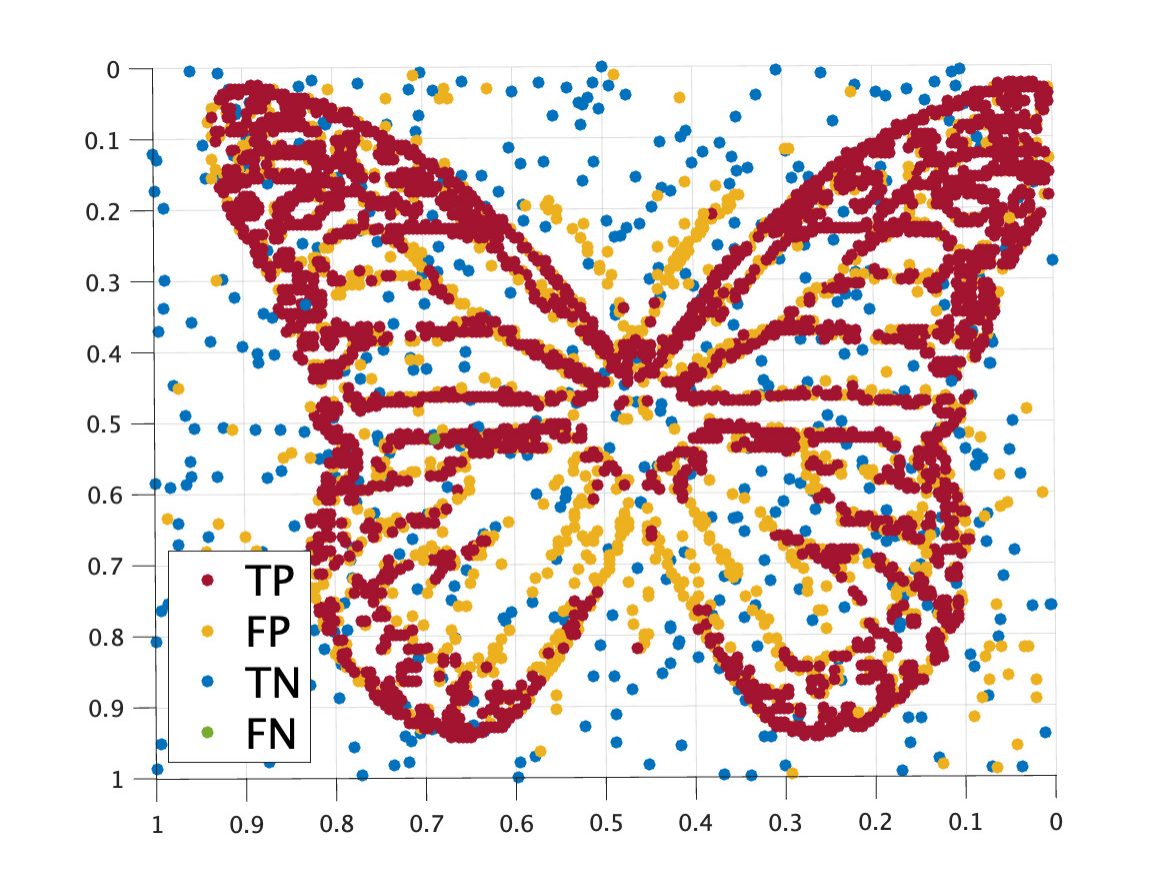} \label{fig:Feng_bike}} 
\subfigure[Denoised: Density Filter]{\includegraphics[width = 0.48\columnwidth]{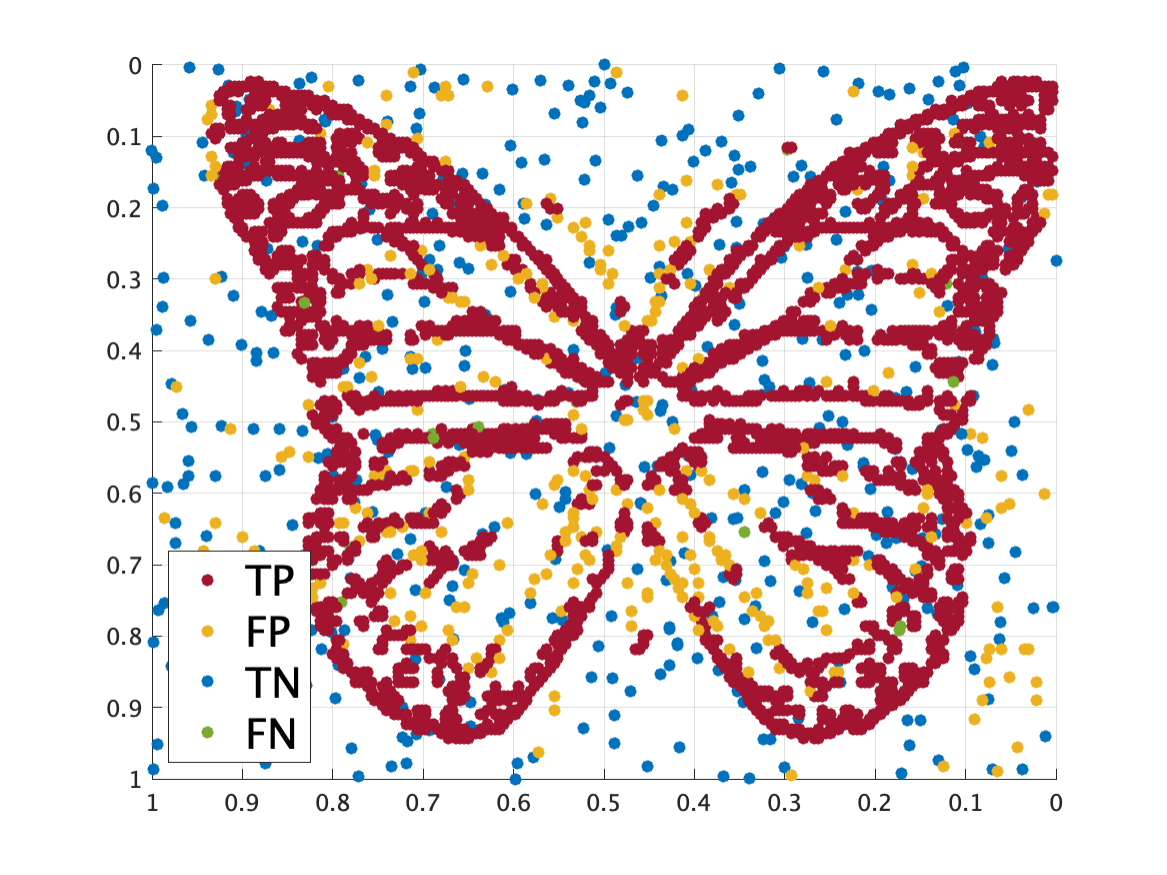} \label{fig:Density_bike}}
\subfigure[Denoised: GNN-Transformer]{\includegraphics[width = 0.48\columnwidth]{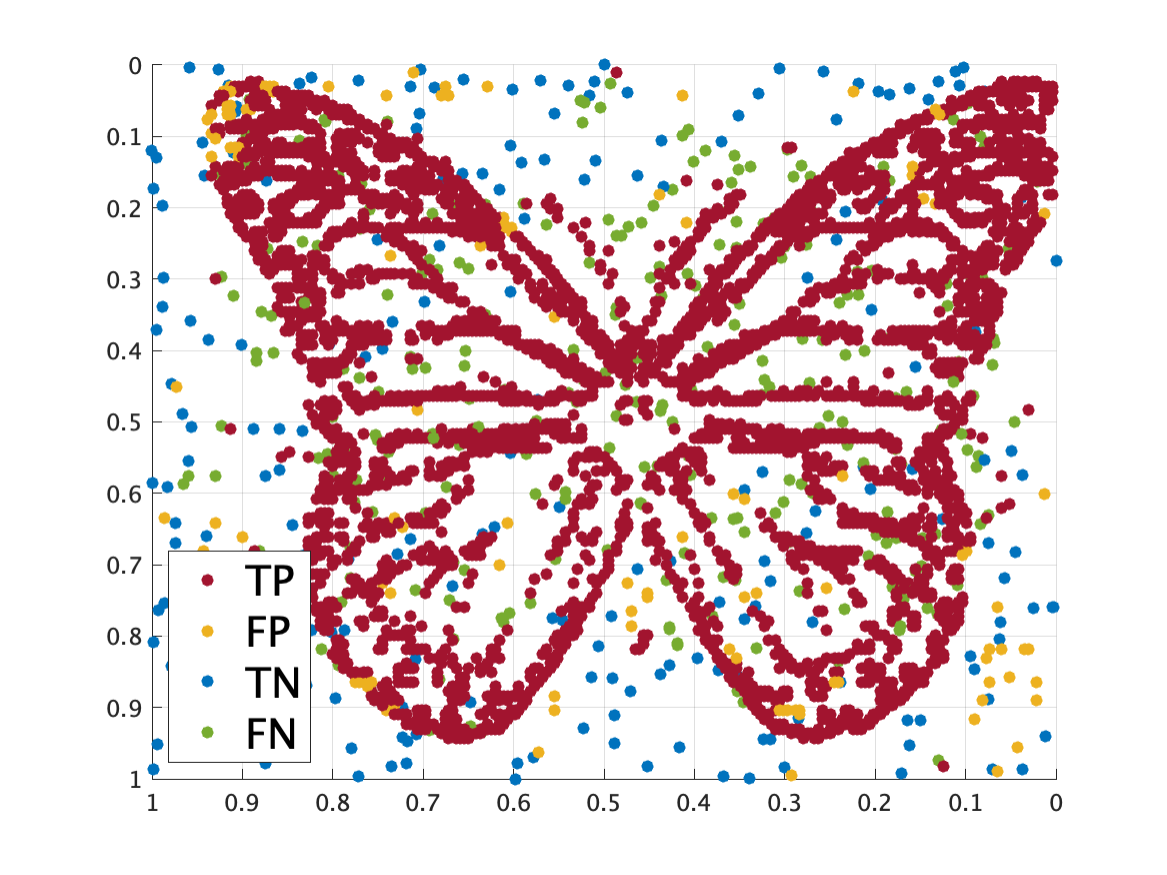} \label{fig:KoGTL_bike}}
\centering
\caption[Scatter plots of denoised \textit{butterfly} in the N-Caltech101 dataset.]{Scatter plots of denoised \textit{butterfly} in the N-Caltech101 dataset. Red points: detected real events (TP). Orange points: undetected real events (FP). Blue points: detected noise events (TN). Green points: undetected noise events (FN).}
\label{fig:synthetic_image}
\end{figure*}

\subsubsection{Setup}\label{sub:setups}
We perform denoising experiments with the following two datasets.\\
\noindent\textbf{N-Caltech101 dataset \cite{orchardConvertingStaticImage2015}:}
    This synthetic dataset consists of event streams capturing $101$ different objects.
    These event streams are synthesized from static 2-D images. 
    We select seven objects from the dataset. In each of these objects, the real events form a \textit{single cluster}.
    Since the dataset is not annotated, we label all original events as real events and artificially add two types of noise events to these events: BA noise is represented by uniformly distributed events, and hot pixel noise is simulated by events continuously distributed along the temporal direction. 
    
    To evaluate the effectiveness of the proposed method under varying noise conditions, we synthesize three event streams where the total noise rate is fixed at $12$\% of the original events.
    Specifically, the noise ratios, i.e., (BA noise, hot pixel noise), are set to $(6\%,6\%)$, $(8\%,4\%)$, and $(10\%,2\%)$, respectively\footnote{In practice, the number of BA noise is typically greater than that of hot pixel noise \cite{gallegoEventbasedVisionSurvey2020}.}.

    \medskip
    
\noindent\textbf{ED-KoGTL dataset \cite{alkendiNeuromorphicCameraDenoising2022}:}
     This is an annotated real-world dataset. It records an event stream of a single moving object for five seconds.
     The annotation is carried out by referencing frame-based images synchronized with the event-based camera.
In the graph construction introduced in Section~\ref{sub:estimate_e}, we normalize the time coordinate in the event stream with $\beta=50$.
After constructing the $\epsilon$NG, we use the RBF kernel in \eqref{eq:epsilon-graph} for edge weights with the decay factor $\gamma=\epsilon/2$. 

For eigenvector calculation in \eqref{eq:modified_eig}, we use the approximation order $\omega=30$ and the number of iterations $k=50$.
We also perform the naive eigenvector calculation.
For an alternative approach, we also perform LOBPCG-based eigenvector calculation \cite{knyazevOptimalPreconditionedEigensolver2000}.
The native calculation and LOBPCG use a graph operator derived from \eqref{eq:modified_eig} to compute eigenvectors.

We compare the denoising performance of the proposed method with three existing methods: 1) spatial filter proposed in \cite{fengEventDensityBased2020a} (abbreviated as Centroid Filter), 2) spatial filter proposed in \cite{zhangNeuromorphicImagingDensitybased2023} (abbreviated as Density Filter), and 3) GNN-driven Transformer (abbreviated as GNN-Transformer) \cite{alkendiNeuromorphicCameraDenoising2022}.

We use four metrics for performance comparison: computation time (CT)\footnote{The experiments using KoGTL have been conducted in Python while the other experiments have been in MATLAB.}, true positive rate (TPR), true negative rate (TNR), and accuracy (Acc).\\
\noindent\textbf{Computational time (CT):}
    We measure the computational time for event denoising. In the proposed method, it also counts on the computational time for the preparation phase in \eqref{eq:modified_eig}. 
    In the GNN-Transformer, we do not include that for the training phase.
    All the experiments are conducted on a MacBook Air with an Apple M2 chip, macOS Ventura 13.0, and 24 GB of RAM.
    \medskip
    
\noindent\textbf{True positive rate (TPR):}
    This metric represents the ratio of the number of the detected real events (TP: true positive) to that of entire real events.
    Denoting the number of the undetected real events (FP: false positive), TPR is measured by $\text{TPR} = \frac{\text{TP}}{\text{TP} + \text{FP}}$.
    \medskip
    
\noindent\textbf{True negative rate (TNR):}
    This metric represents the ratio of the number of the detected noise events (TN: true negative) to that of entire noise events.
    Denoting the number of the undetected noise events (FN: false negative), TNR is measured by $\text{TNR} = \frac{\text{TN}}{\text{TN} + \text{FN}}$.
    \medskip
    
\noindent\textbf{Accuracy (Acc):}
    This metric represents the ratio of the number of correctly detected events to that of entire events.
    Accuracy is given by $\text{Accuracy} = \frac{\text{TP} + \text{TN}}{\text{TP} + \text{TN} + \text{FP} + \text{FN}}$.



\begin{figure*}[t]
\centering
\subfigure[Ground truth]{\includegraphics[width = 0.48\columnwidth]{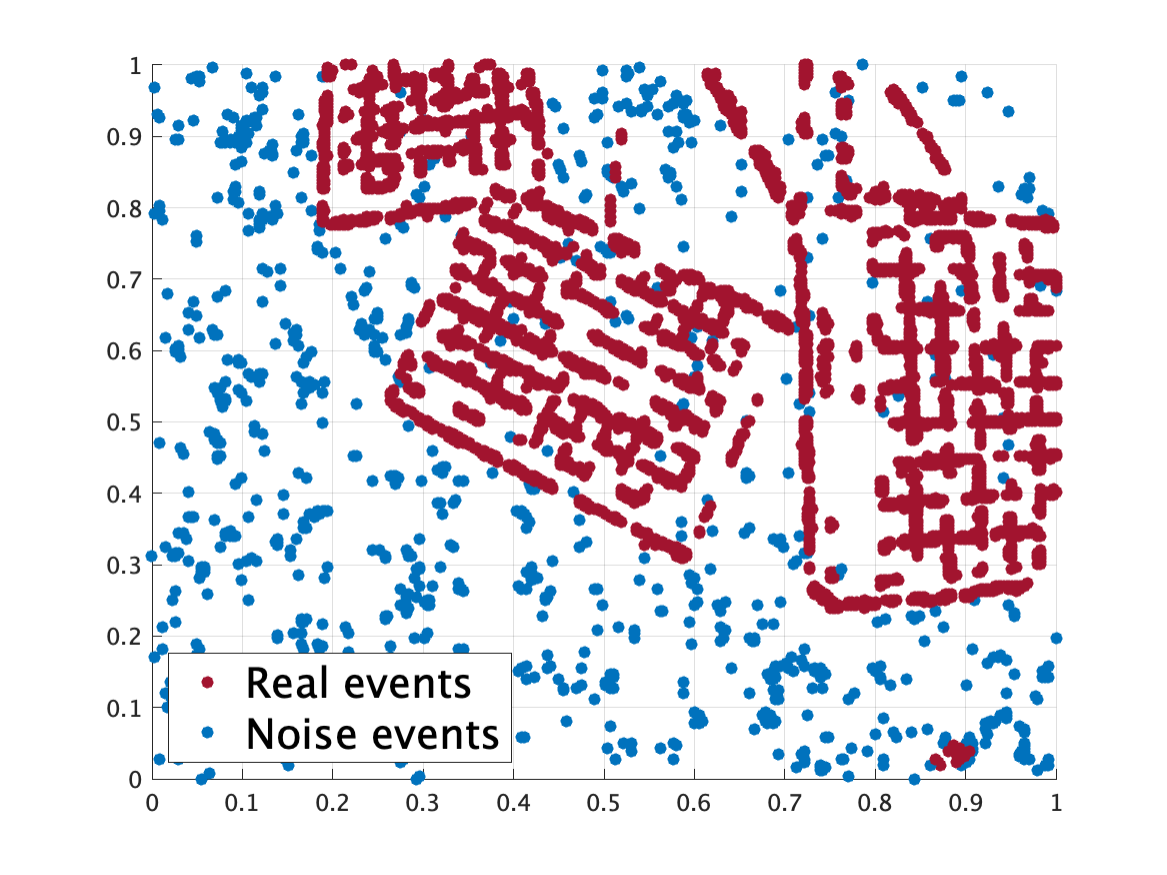}\label{fig:groundtruth}} 
\subfigure[Denoised: proposed w/ EVD]{\includegraphics[width = 0.48\columnwidth]{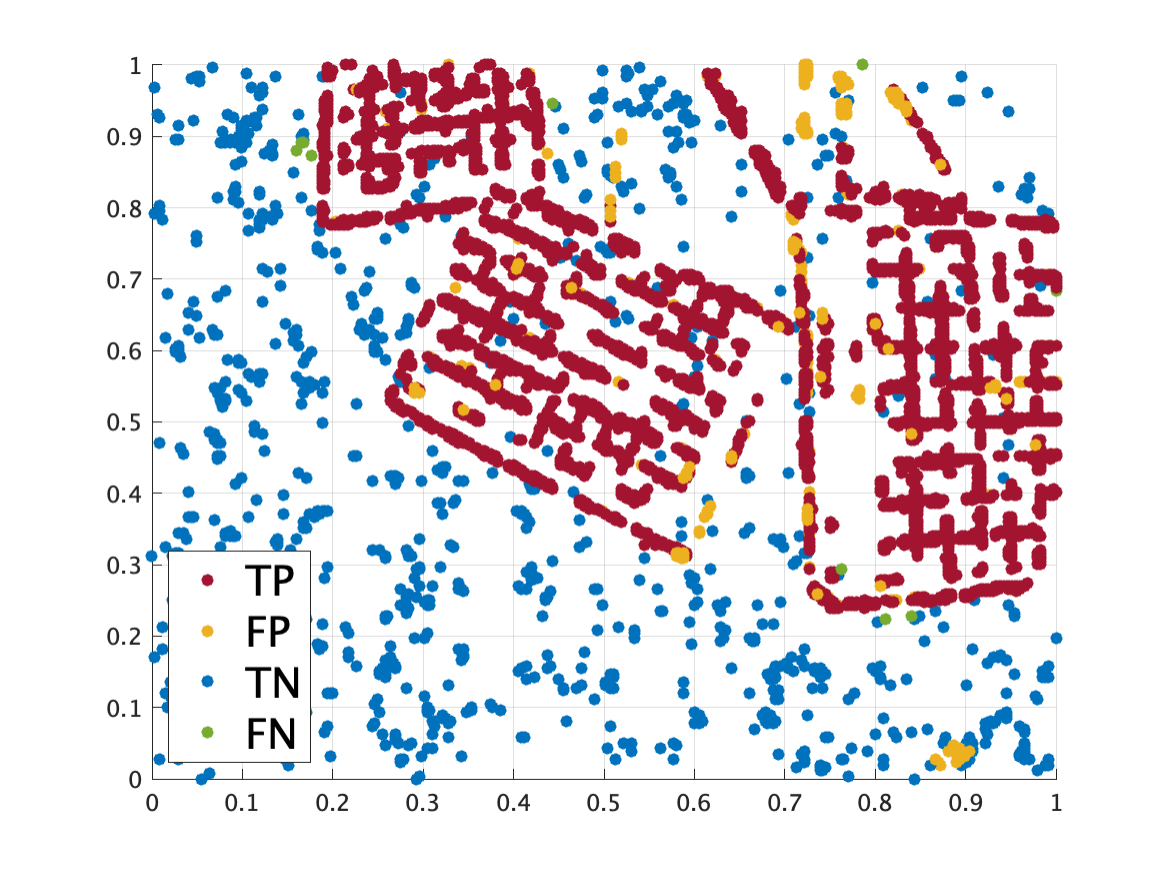}\label{fig:epsilon}}
\subfigure[Denoised: proposed w/ Power Method]{\includegraphics[width = 0.48\columnwidth]{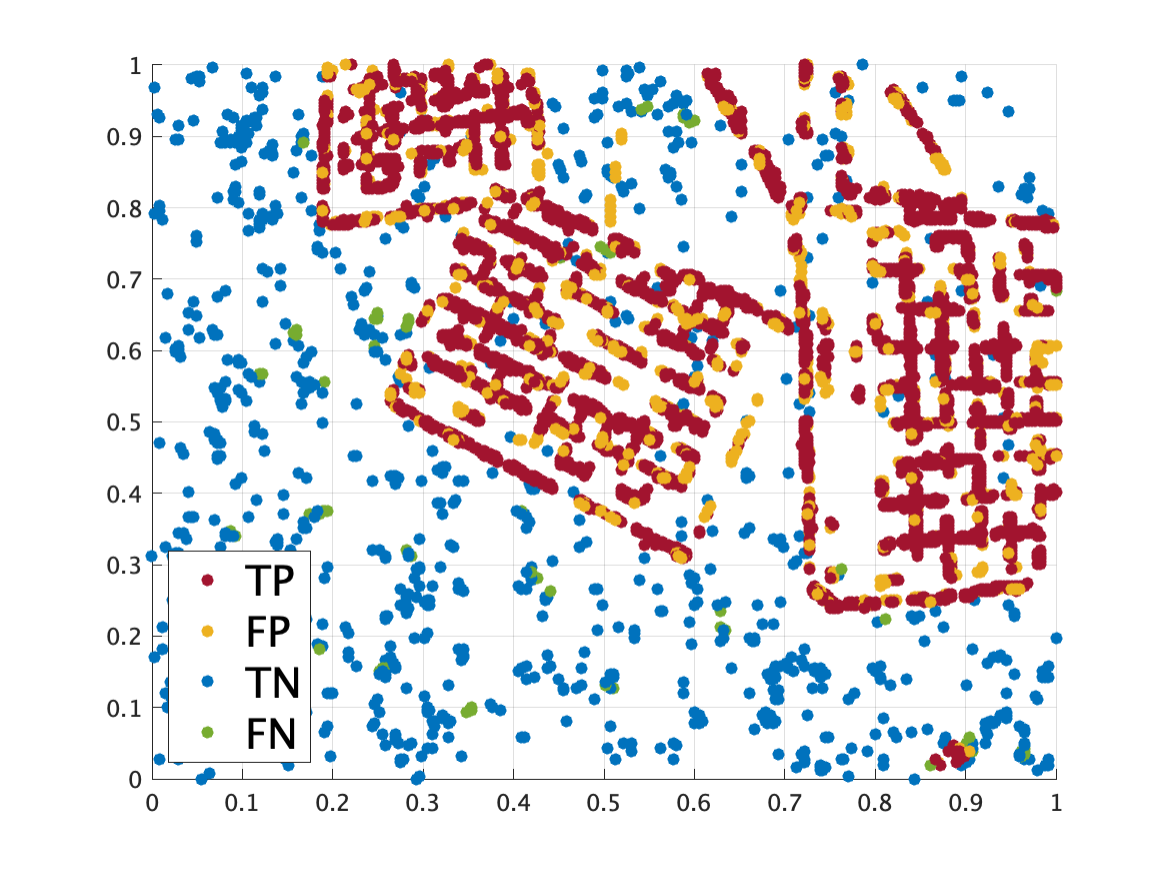}\label{fig:power}}  
\subfigure[Denoised: proposed w/ LOBPCG]{\includegraphics[width = 0.48\columnwidth]{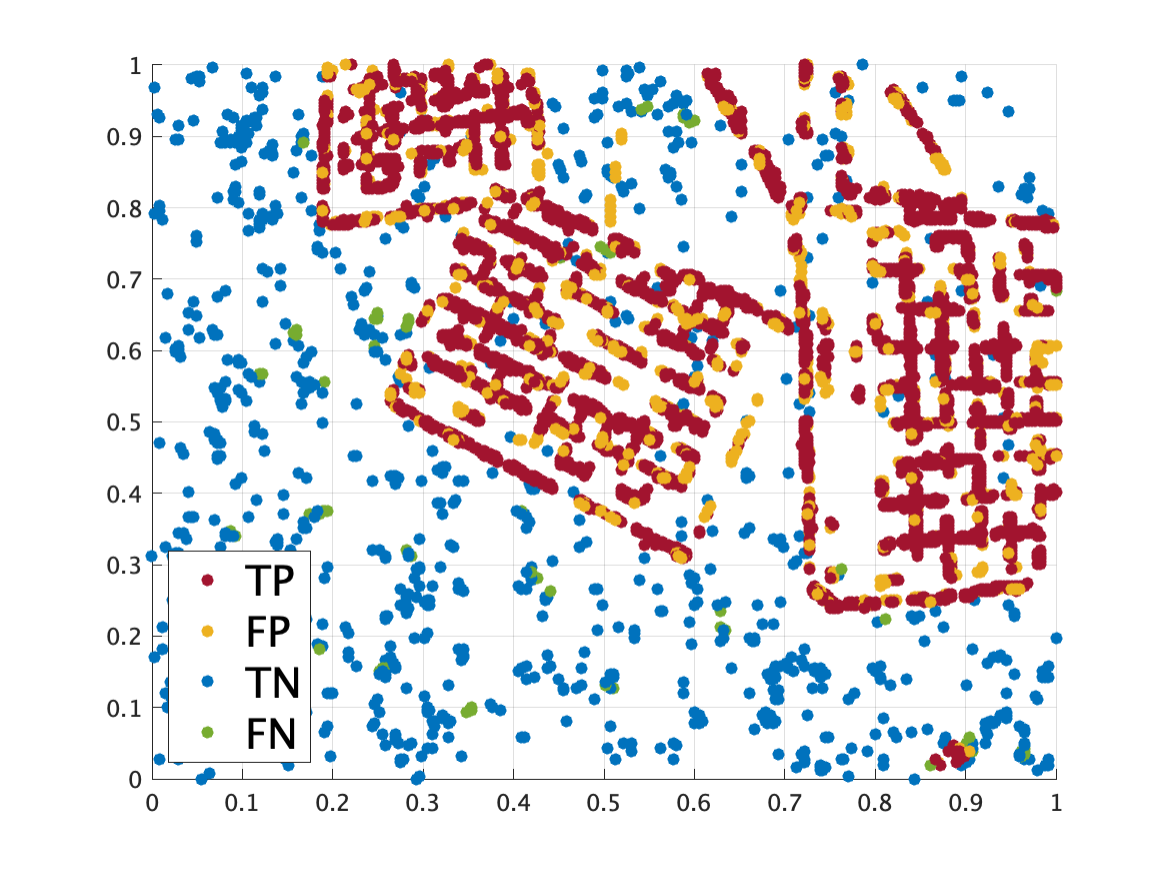}\label{fig:lobpcg}}\\
\subfigure[Denoised: Centroid Filter]{\includegraphics[width = 0.48\columnwidth]{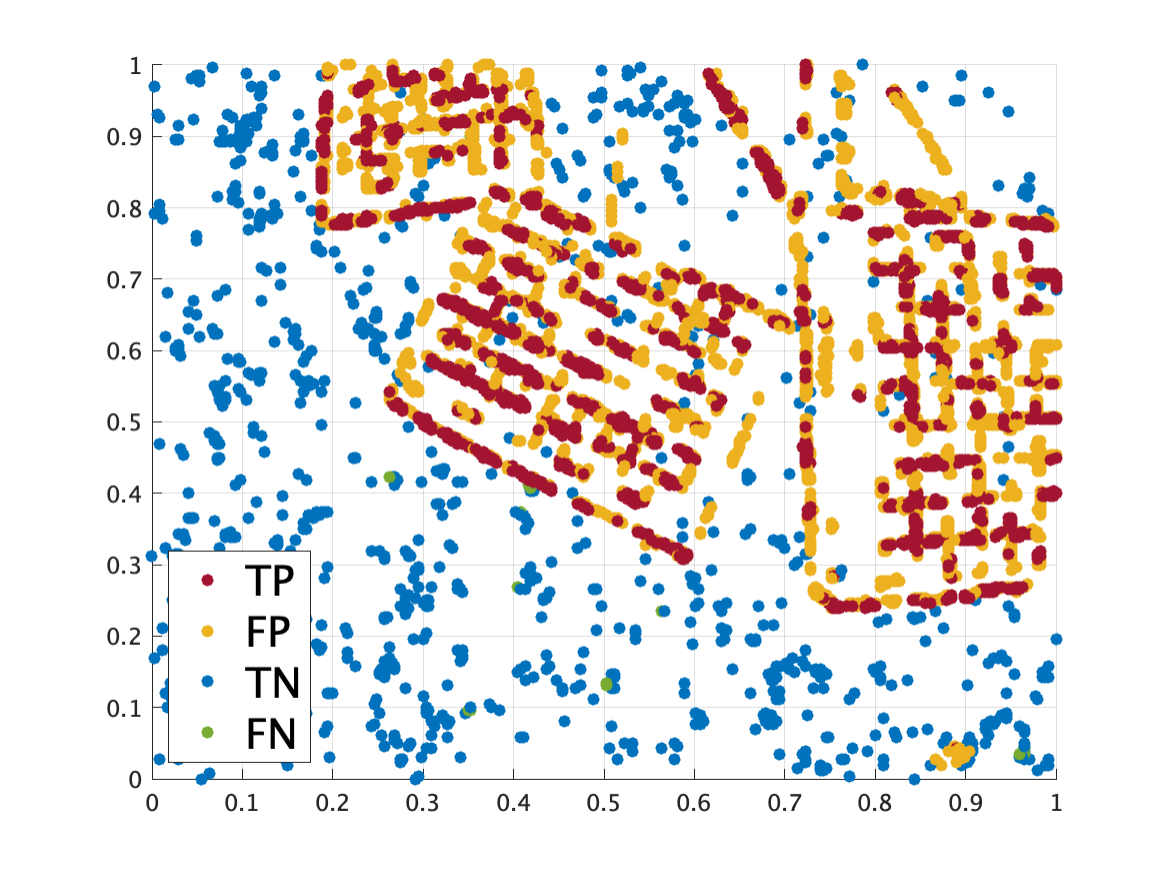}\label{fig:Feng}}
\subfigure[Denoised: Density Filter]{\includegraphics[width = 0.48\columnwidth]{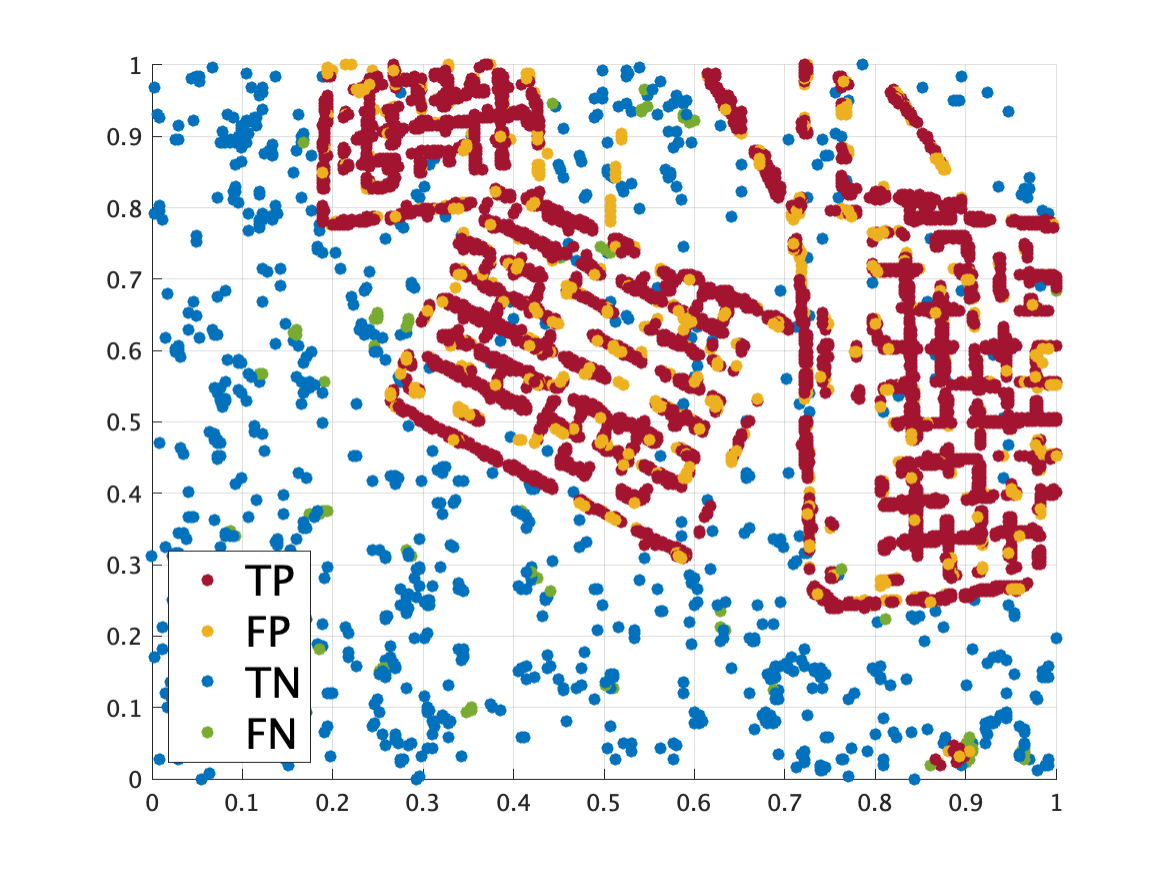}\label{fig:Density}}
\subfigure[Denoised: GNN-Transformer]{\includegraphics[width = 0.48\columnwidth]{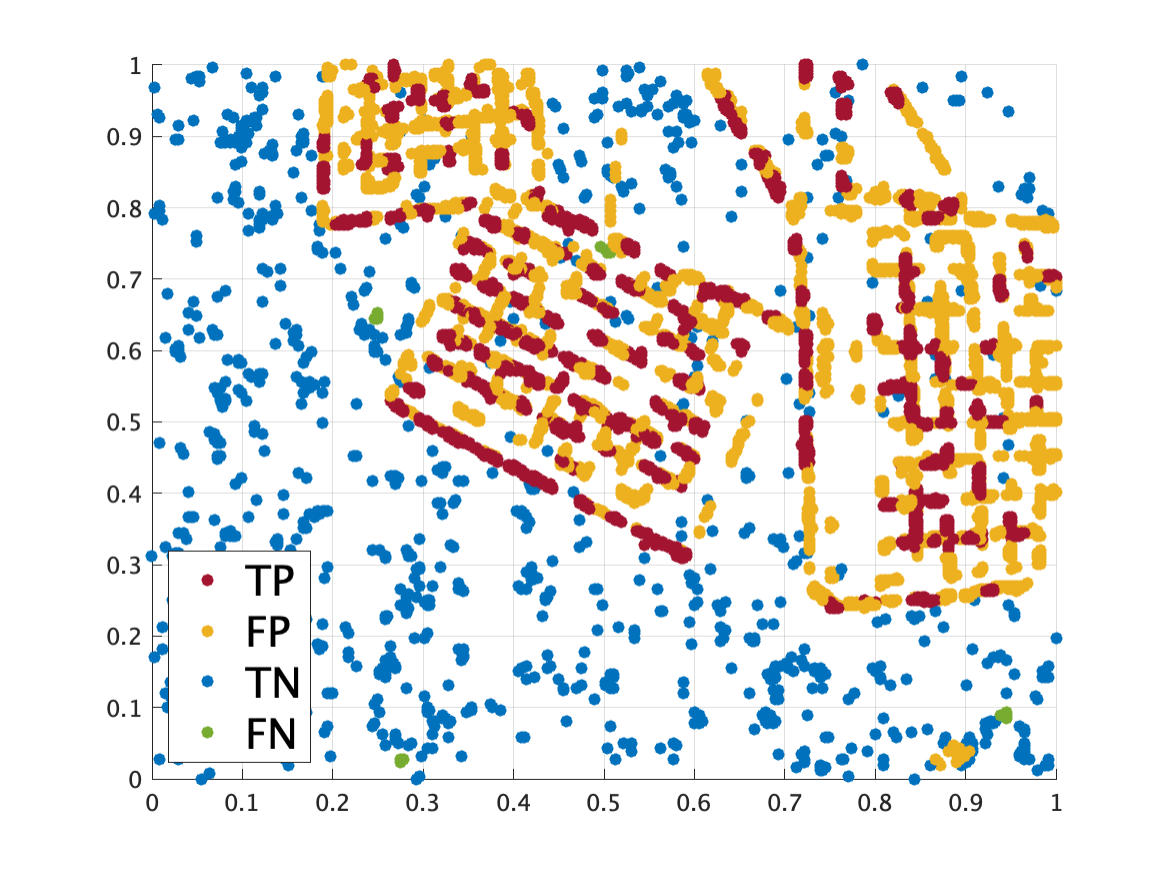}\label{fig:KoGTL}} 
\centering
\caption[Scatter plots of denoised events in the ED-KoGTL dataset.]{Scatter plots of denoised events in the ED-KoGTL dataset. Events are colored in the same manner as in Fig.~\ref{fig:synthetic_image}.}
\label{fig:real_image}
\end{figure*}


\subsubsection{Results: Qualitative Evaluation}\label{sub:result_single}
 Examples of the denoised events are visualized in Figs.~\ref{fig:synthetic_image} and \ref{fig:real_image}. We obtain these visualizations by mapping denoised events in 2-D spatial coordinates.
 
\noindent\textbf{N-Caltech101 dataset:}
    As observed in Fig.~\ref{fig:synthetic_image}, the proposed method with EVD successfully removes noise events. It indicates that the $\epsilon$NG can properly reflect the density prior.
    The proposed method implemented with the power method and LOBPCG exhibits comparable results with its naive implementation.
    While the Density Filter seems to show the highest performance among existing denoising methods, real events still remain undetected. 
    The Centroid Filter often removes real events as noise events accidentally.
    The noise events still remain for event denoising using GNN-Transformer.
    \medskip
    
\noindent\textbf{ED-KoGTL dataset:}
    As observed in Fig.~\ref{fig:real_image}, the proposed method, especially one with EVD, successfully removes noise events.
    The Centroid Filter and GNN-Transformer classify many real events as noise.

\subsubsection{Results: Quantitative Evaluation}
The experimental results on N-Caltech101 and ED-KoGTL are summarized in Table~\ref{tab:tab1}.

\noindent\textbf{N-Caltech101 dataset:}
    In Table~\ref{tab:tab1}\ref{tab:sub1}, the proposed method implemented with the customized power method and LOBPCG is comparable with its naive implementation in TPR, TNR, and Acc, while CT is significantly faster.
    This indicates that the combination of the transformation \eqref{eq:modified_eig} and fast eigensolvers works properly.
    
    The proposed method consistently achieves high Acc and TNR across varying noise ratios. In contrast, alternative methods have low TNR, as the level of hot pixel noise increases. This indicates that the proposed method is robust against ratios of noise events.
    The Centroid Filter shows high TNR but low TPR and Acc for all noise ratios. This suggests that it can mistakenly remove real events as noise events.
    Although GNN-Transformer shows high TPR and Accuracy for all noise ratios, it suffers from low TNR. This implies that it misclassifies many noise events as real ones.
    \medskip

\noindent\textbf{ED-KoGTL dataset:}
    As observed in Table~\ref{tab:tab1}\ref{tab:sub2}, the proposed method with EVD outperforms existing methods for all metrics except for CT. This indicates that the $\epsilon$NG suitably captures the irregular structure of the event stream as in the previous example. 
    Furthermore, we observe that the proposed method implemented with the power method and LOBPCG is comparable with its naive implementation in Acc while CT is significantly faster.

\begin{table}[tp]
    \caption{Experimental Results for Single Clustered Scenarios. (a) N-Caltech101 dataset (average for seven objects). 
    (b) ED-KoGTL dataset. Bold numbers indicate the best results while underlined numbers represent the second-best results.}
    \label{tab:tab1}
    \centering
    
    \subfigure[N-Caltech101 dataset (average for seven objects)]{%
    \label{tab:sub1}%
    \begin{tabular}{c|c|c|c|c}\hline
        & CT [sec] & TPR & TNR & Acc \\ \hline\hline
        \multicolumn{5}{c}{BA noise: 10\%, Hot pixel noise: 2\%} \\ \hline
        Proposed w/ EVD & $45.12$ & $0.915$ & \underline{0.992} & 0.924\\ 
        Proposed w/ Power Method & $7.15$ & \underline{0.942} & $0.894$ & \textbf{0.936}\\ 
        Proposed w/ LOBPCG & $5.2$ & $0.931$ & $0.898$ & $\underline{0.927}$\\ \hline
        Centroid Filter & $110$ & $0.734$ & \textbf{0.999} & $0.763$\\ 
        Density Filter & \underline{0.93} & $0.921$ & $0.790$ & $0.906$\\  
        GNN-Transformer & \textbf{0.62} & \textbf{0.959} & $0.456$ & $0.902$\\ \hline\hline
        
        \multicolumn{5}{c}{BA noise: 8\%, Hot pixel noise: 4\%} \\ \hline
        Proposed w/ EVD & $44.38$ & \underline{0.946} & \underline{0.976} & \textbf{0.947}\\ 
        Proposed w/ Power Method & $7.52$ & $0.943$ & $0.848$ & \underline{0.932}\\ 
        Proposed w/ LOBPCG & $5.1$ & $0.935$ & $0.876$ & $0.928$\\ \hline
        Centroid Filter & $107$ & $0.706$ & \textbf{0.999} & $0.738$\\ 
        Density Filter & \underline{0.84} & $0.912$ & $0.754$ & $0.894$\\  
        GNN-Transformer & \textbf{0.39} & \textbf{0.963} & $0.395$ & $0.901$\\ \hline\hline
        
        \multicolumn{5}{c}{BA noise: 6\%, Hot pixel noise: 6\%} \\ \hline
        Proposed w/ EVD & $43.24$ & $0.945$ & \underline{0.990} & \textbf{0.950}\\ 
        Proposed w/ Power Method & $5.50$ & \underline{0.950} & $0.892$ & \underline{0.944}\\ 
        Proposed w/ LOBPCG & $5.3$ & $0.940$ & $0.865$ & $0.931$\\ \hline
        Centroid Filter & $106$ & $0.729$ & \textbf{0.999} & $0.759$\\ 
        Density Filter & \underline{0.86} & $0.894$ & $0.702$ & $0.873$\\  
        GNN-Transformer & \textbf{0.60} & \textbf{0.963} & $0.318$ & $0.890$\\ \hline
    \end{tabular}
    }

    \subfigure[ED-KoGTL dataset]{%
    \label{tab:sub2}%
    \begin{tabular}{c|c|c|c|c}\hline
        & CT [sec] & TPR & TNR & Acc \\ \hline\hline
        Proposed w/ EVD & $171$ & \textbf{0.943} & \textbf{0.995} & \textbf{0.956}\\ 
        Proposed w/ Power Method & $3.60$ & \underline{0.919} & $0.862$ & \underline{0.905}\\ 
        Proposed w/ LOBPCG & $2.8$ & $0.911$ & $0.872$ & $0.902$\\ \hline
        Centroid Filter & $424$ & $0.694$ & $0.911$ & $0.747$\\ 
        Density Filter & \underline{1.62} & $0.877$ & $0.792$ & $0.856$\\ 
        GNN-Transformer & \textbf{0.35} & $0.405$ & \underline{0.990} & $0.548$\\ \hline
    \end{tabular}
    }
\end{table}

\subsection{Denoising with Multiple Clusters of Real Events}

\begin{figure*}[t]
\centering
\subfigure[Ground truth]{\includegraphics[width = 0.48\columnwidth]{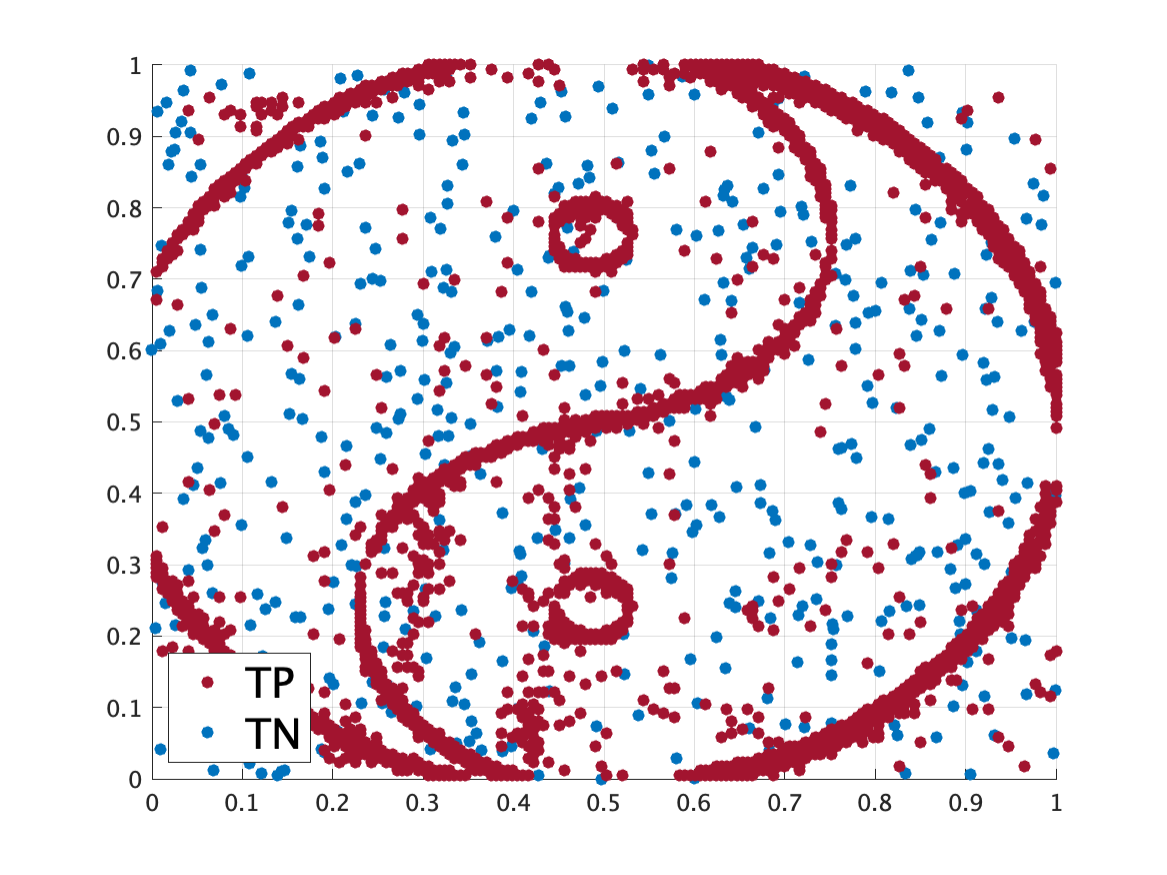} \label{fig:groundtruth_mul}} 
\subfigure[Denoised: proposed w/ EVD]{\includegraphics[width = 0.48\columnwidth]{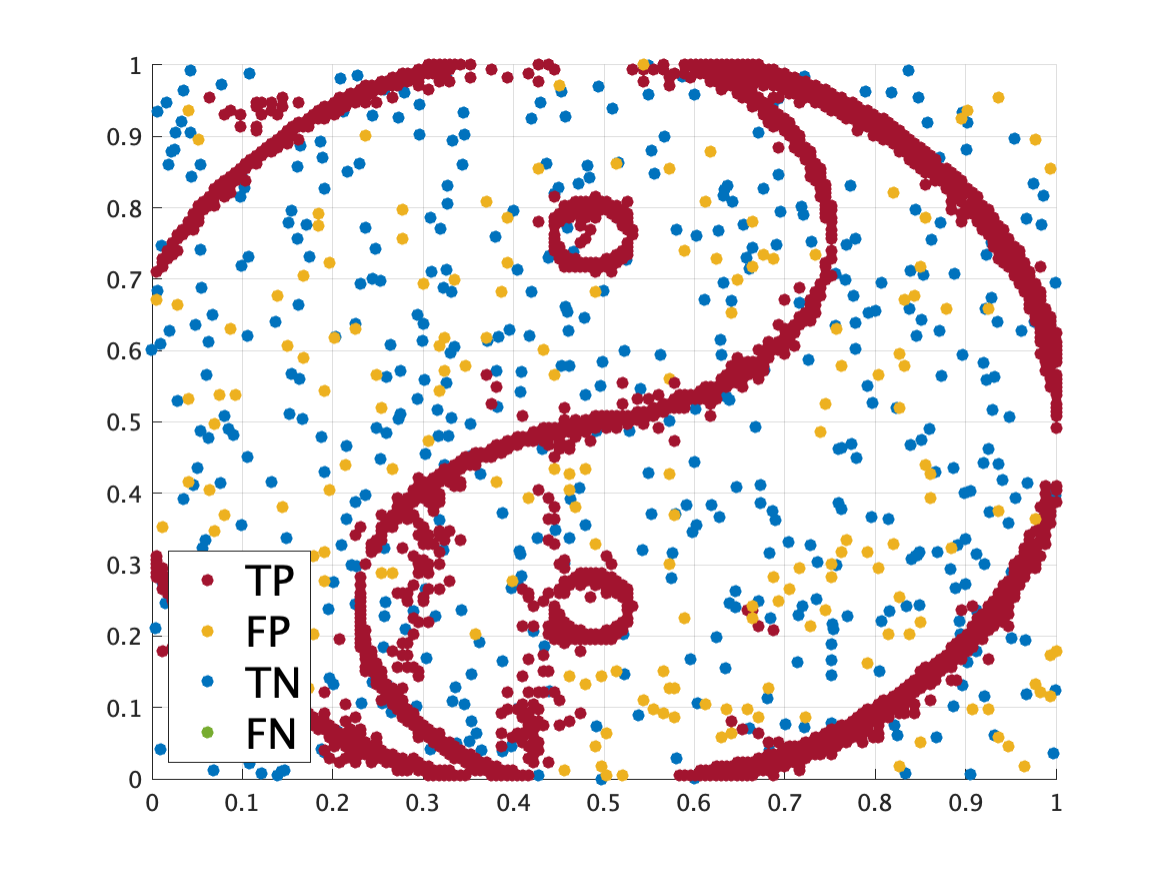} \label{fig:mul_epsilon}} 
\subfigure[Denoised: proposed w/ Power Method]{\includegraphics[width = 0.48\columnwidth]{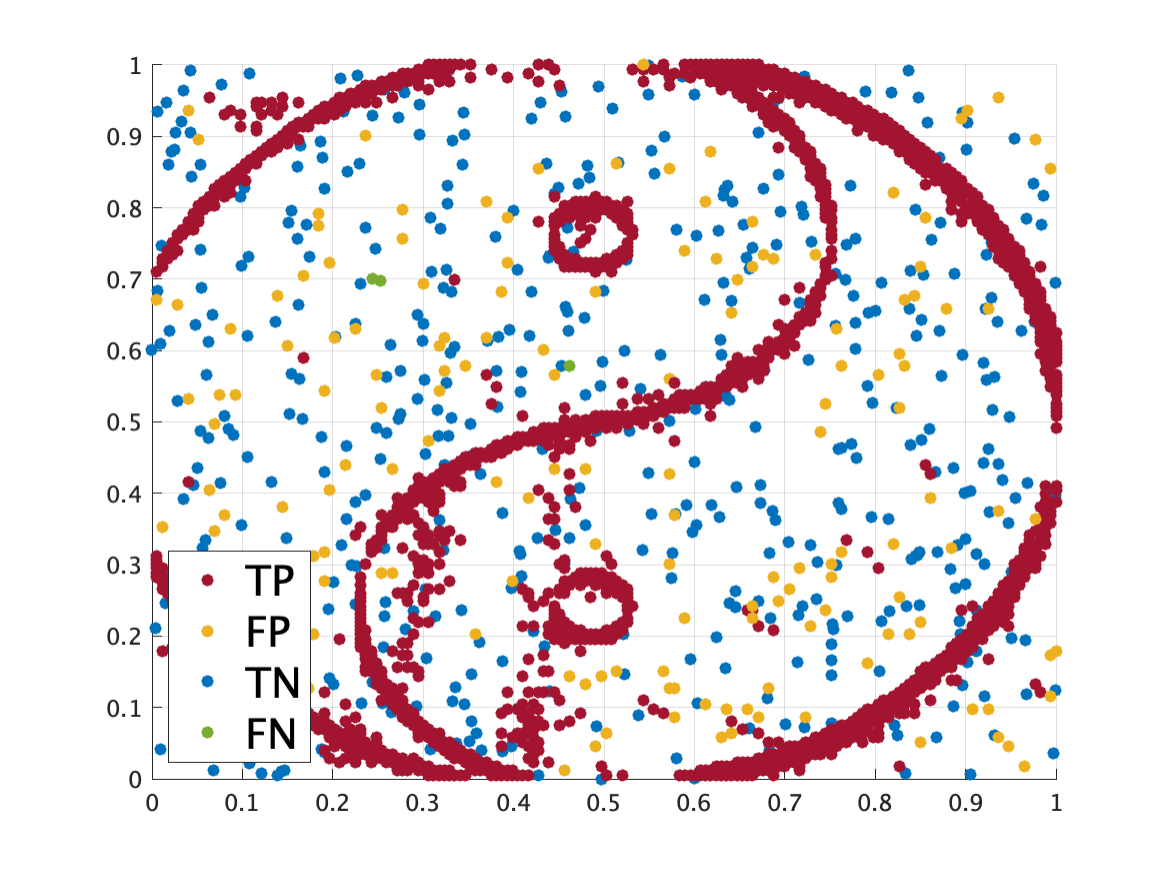} \label{fig:mul_power}} 
\subfigure[Denoised: proposed w/ LOBPCG]{\includegraphics[width = 0.48\columnwidth]{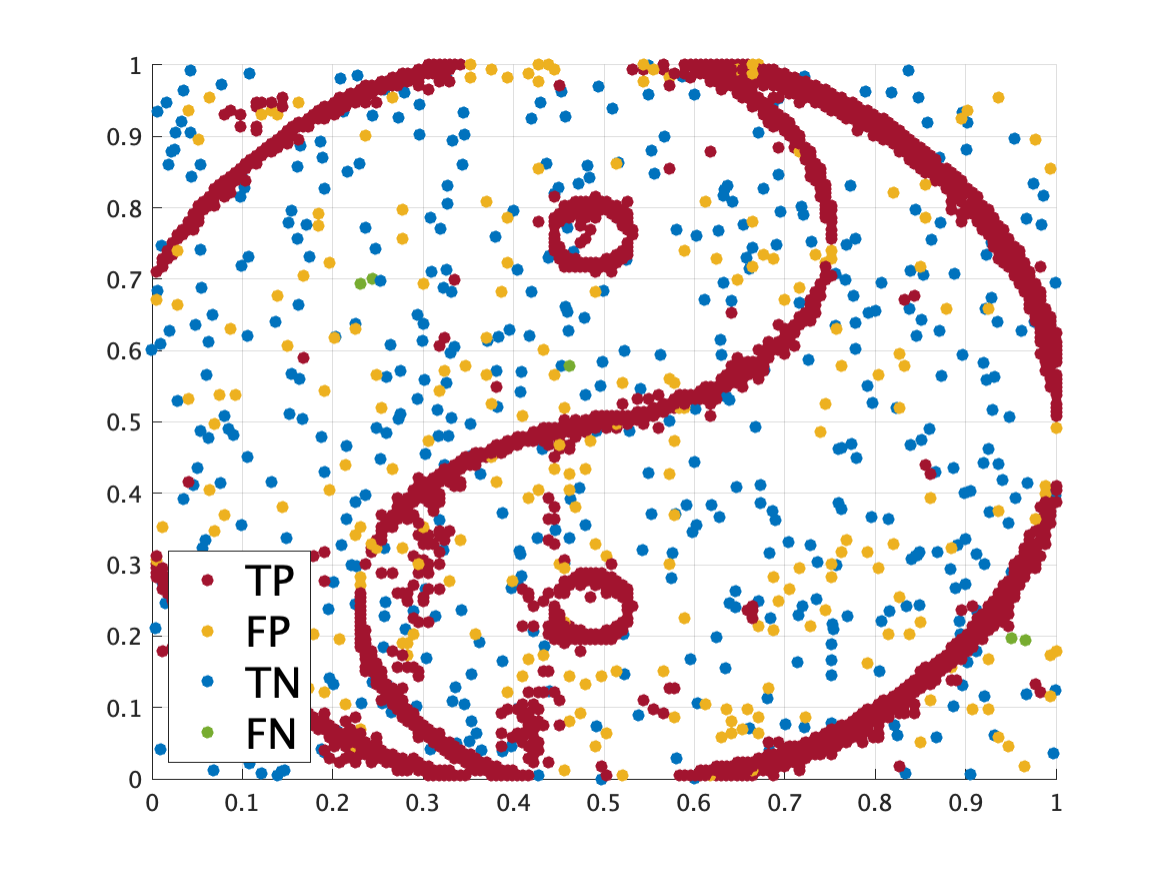} \label{fig:mul_lopcg}}
\subfigure[Denoised: Centroid Filter]{\includegraphics[width = 0.48\columnwidth]{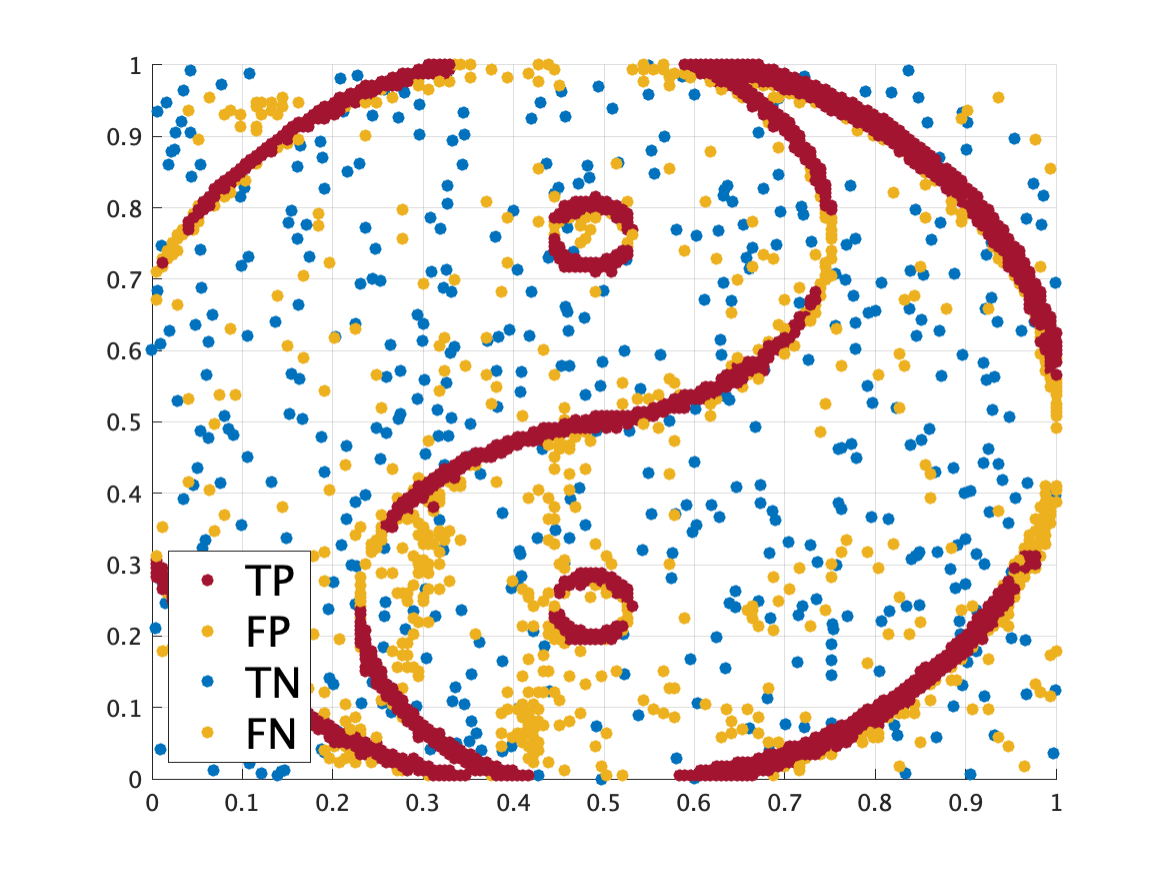} \label{fig:mul_feng}}
\subfigure[Denoised: Density Filter]{\includegraphics[width = 0.48\columnwidth]{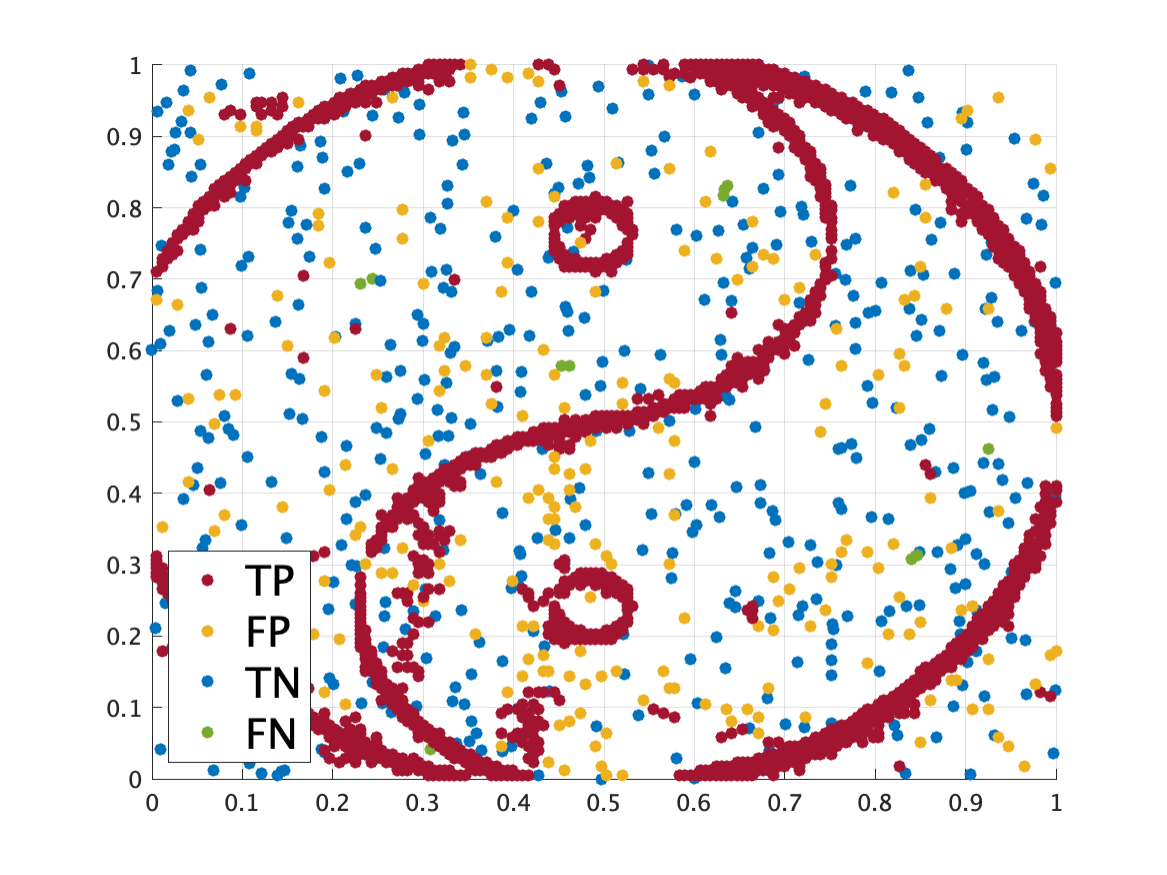} \label{fig:mul_density}}
\subfigure[Denoised: GNN-Transformer]{\includegraphics[width = 0.48\columnwidth]{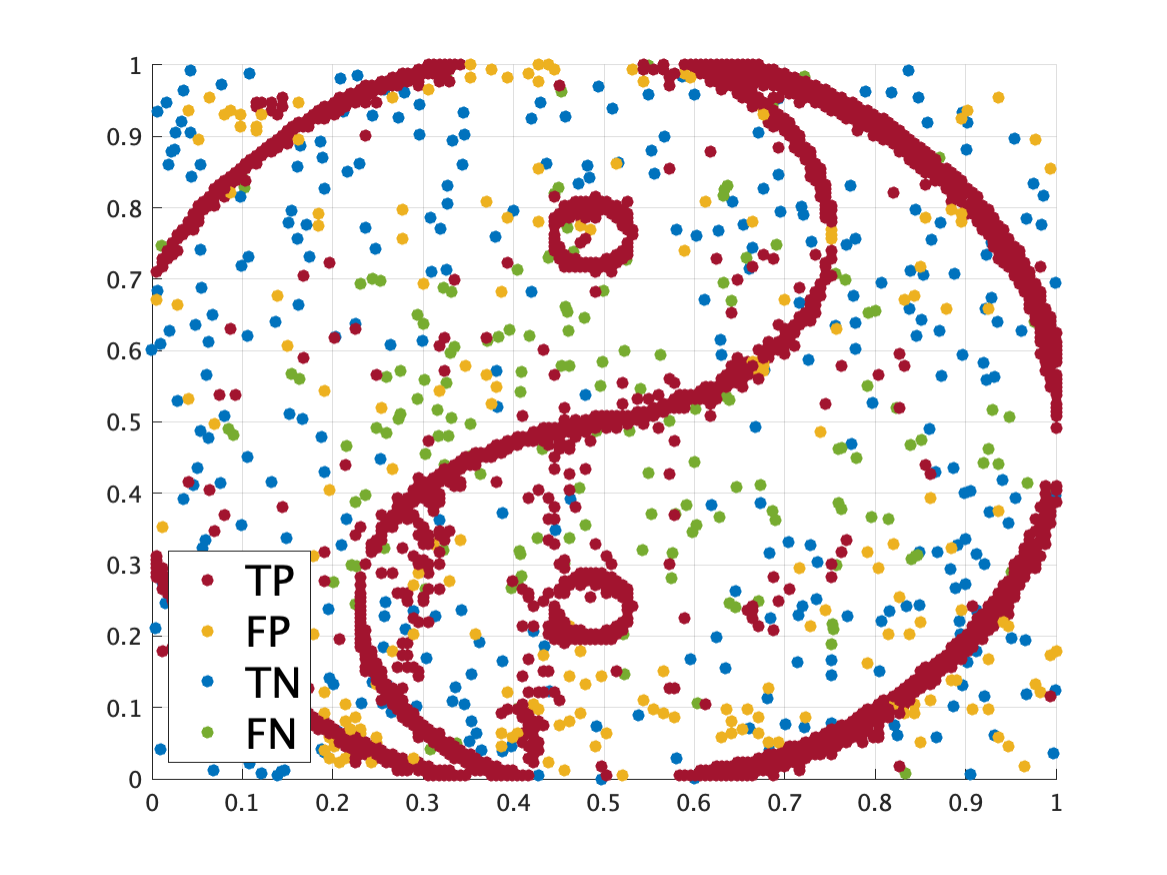} \label{fig:mul_kogtl}}
\centering
\caption{Scatter plots of denoised events. Events are colored in the same manner as in Fig.~\ref{fig:synthetic_image}.}
\label{fig:denoised_image}
\end{figure*}

Here, we conduct denoising experiments with real-world events forming multiple clusters.

\subsubsection{Setup}

We use the N-Caltech101 dataset introduced in Section~\ref{sub:setups}. We conduct denoising experiments on an event stream that forms \textit{multiple clusters}.
All other experimental settings remain the same as in Section~\ref{sub:setups}.


\subsubsection{Results: Qualitative Evaluation}

We visualize the 2-D image of denoised events in Fig.~\ref{fig:denoised_image}.
As shown in Fig.~\ref{fig:denoised_image}\ref{fig:mul_epsilon}, the proposed method with EVD effectively detects noise events. 
This indicates that the $\epsilon$NG can also capture the dense nature of multiple clusters formed by real events.
Similar to Section~\ref{sub:result_single}, the Centroid Filter tends to misclassify real events as noise while the GNN-Transformer often fails to remove noise events.

\subsubsection{Results: Quantitative Evaluation}
Experimental results are summarized in Table~\ref{tab:tab_multi}. As shown in the Table~\ref{tab:tab_multi}, the proposed method implemented with the power method and LOBPCG is comparable with its naive implementation in terms of TPR, TNR, and Acc, while it significantly reduces CT.
This indicates that the combination of the transformation \eqref{eq:modified_eig} and fast eigensolvers works properly.
Furthermore, the proposed method displays the same tendencies as observed in Section~\ref{sub:result_single} when compared to alternative methods, often outperforming alternative methods in all metrics.


\begin{table}[tp]
    \caption{Experimental Results for Multiple Clustered Scenarios. Bold numbers and underlined numbers represent the best results and the second-best results, respectively.}
    \centering
    \begin{tabular}{c|c|c|c|c}\hline
        & CT [sec] & TPR & TNR & Acc \\ \hline\hline
      \multicolumn{5}{c}{BA noise: 10\%, Hot pixel noise: 2\%} \\ \hline
      Proposed w/ EVD & $44.20$ & $0.954$ & $\underline{0.963}$ & $\mathbf{0.955}$\\ 
      Proposed w/ Power Method & $1.88$ & $\mathbf{0.965}$ & $0.872$ & $\underline{0.954}$\\ 
      Proposed w/ LOBPCG & $\underline{1.39}$ & $0.939$ & $0.878$ & $0.932$\\ \hline
      Centroid Filter & $107$ & $0.809$ & $\mathbf{1}$ & $0.831$\\ 
      Density Filter & $1.51$ & $0.951$ & $0.764$ & $0.930$\\  
      GNN-Transformer & $\mathbf{0.77}$ & $\underline{0.955}$ & $0.520$ & $0.905$\\ \hline\hline
      \multicolumn{5}{c}{BA noise: 8\%, Hot pixel noise: 4\%} \\ \hline
      Proposed w/ EVD & $42.81$ & $\underline{0.936}$ & $\underline{0.943}$ & $\mathbf{0.937}$\\ 
      Proposed w/ Power Method & $2.51$ & $\underline{0.936}$ & $0.869$ & $\underline{0.928}$\\ 
      Proposed w/ LOBPCG & $\underline{0.82}$ & $0.908$ & $0.880$ & $0.905$\\ \hline
      Centroid Filter & $102$ & $0.806$ & $\mathbf{1}$ & $0.828$\\ 
      Density Filter & $0.85$ & $0.920$ & $0.701$ & $0.896$\\  
      GNN-Transformer & $\mathbf{0.64}$ & $\mathbf{0.955}$ & $0.487$ & $0.905$\\ \hline\hline
      \multicolumn{5}{c}{BA noise: 6\%, Hot pixel noise: 6\%} \\ \hline
      Proposed w/ EVD & $42.77$ & $0.936$ & $\underline{0.966}$ & $\mathbf{0.940}$\\ 
      Proposed w/ Power Method & $2.55$ & $\underline{0.938}$ & $0.869$ & $\underline{0.929}$\\ 
      Proposed w/ LOBPCG & $\underline{0.81}$ & $0.904$ & $0.865$ & $0.900$\\ \hline
      Centroid Filter & $106$ & $0.807$ & $\mathbf{1}$ & $0.828$\\ 
      Density Filter & $0.85$ & $0.907$ & $0.678$ & $0.883$\\  
      GNN-Transformer & $\mathbf{0.37}$ & $\mathbf{0.955}$ & $0.445$ & $0.900$\\ \hline\hline
    \end{tabular}
    \label{tab:tab_multi}
\end{table}

\section{Conclusion}\label{sec:conclusion}
This paper proposes a simple yet effective event denoising method based on graph spectral features.
We identify clusters of real events as large connected components in the graph.
In the proposed method, we first construct the $\epsilon$NG with an estimated parameter encapsulating the density prior.
Then, we obtain the eigenvectors corresponding to the small eigenvalues. In the calculation of eigenvectors, we reorder the non-zero eigenvalues of the graph Laplacian in a reverse order. This allows us to utilize the existing fast eigensolver algorithms and thereby reduce the computational complexity.
We detect large connected components by analyzing the non-zero elements of eigenvectors.
In experiments, we demonstrate that the proposed method effectively removes noise events from the raw event streams compared to alternative methods.

\appendix

\section{Comparison of Graph Construction Methods in Our Method}
\label{sec:graph_comparison}

Here, we study the effect of graph construction methods while this paper primarily focuses on the $\epsilon$NG.
We compare the $\epsilon$NG with the following graph construction methods.\\
\noindent\textbf{$\bm{k}$-Nearest Neighbor Graph ($\bm{k}$NNG):}
    The $k$NNG is one of the most common graph construction methods. Each node is connected to its $k$ nearest neighbors.
    
    Let us assume we have a set of signals 
    \(\mathbf{X} = [\mathbf{x}_1, \dots, \mathbf{x}_N]^\top \in \mathbb{R}^{N \times M}\)
    and a parameter $k$. The unweighted adjacency matrix of the $k$NNG, \(\mathbf{A}_{k\text{NNG}}\), is obtained by solving the following problem.
    \begin{equation}
        [\mathbf{A}_{k\text{NNG}}]_{i,:}
        = \argmin_{\mathbf{A}\in\{0,1\}^{N\times N}}
        \|[\mathbf{D}\circ\mathbf{A}]_{i,:}\|_1 \hspace{1em} \text{s.t. }
        \|[\mathbf{A}]_{i,:}\|_1 = k,
    \end{equation}
    where \(\mathbf{D}\) is the pairwise distance matrix for the rows of \(\mathbf{X}\), i.e., $[\mathbf{D}]_{i,j}=\|\mathbf{x}_i-\mathbf{x}_j\|$. This method assigns the fixed $k$ to each node.
    \medskip

\noindent\textbf{Varied-$\bm{k}$-Nearest Neighbor Graph (v$\bm{k}$NNG)\cite{zhangEfficientKNNClassification2018, tamaruOptimizingKNNGraphs2024}:}
    The v$k$NNG extends the $k$NNG by allowing each node to have a variable number of neighbors $k_i$. Its unweighted adjacency matrix \(\mathbf{A}_{\text{V}k\text{NNG}}\) is given by:
    \begin{equation}\label{eq:vknn}
    \begin{gathered}
        [\mathbf{A}_{\text{V}k\text{NNG}}]_{i}=\argmax_{\mathbf{A}\in\{0,1\}^{N\times N}} \|[\mathbf{A}]_{i,:}\|_1 \\
        \text{s.t. }\|[\mathbf{D} \circ \mathbf{A}]_{i,:}\|_1\leq \frac{1}{N}\|[\mathbf{D}]_{i,:}\|_1.
    \end{gathered}
    \end{equation}
    In \eqref{eq:vknn}, $k_i=\|[\mathbf{A}_{\text{V}k\text{NNG}}]_{i,:}\|_1$ is automatically determined based the distance to other nodes, and larger $k_i$ values are assigned to denser nodes.
    In this sense, the v$k$NNG can be viewed as an extension of $\epsilon$NG in which each node has its own radius \(\epsilon_i\).

\subsection{Setup}
We evaluate denoising accuracies on the ED-KoGTL dataset introduced in Section~\ref{sub:setups} as a function of the connectivity of graphs. For a fair comparison, we equalize the total number of edges for the three graph construction methods, and evaluate the denoising accuracy.

\subsection{Results}
Figure~\ref{fig:graph_copmarison} shows the comparison for the three graph construction methods. The $k$NNG is consistent but lower than the others when the graph become dense.
This is because real and noise events are connected with a large $k$.

The v$k$NNG achieves high accuracy when low to middle connectivities, but its accuracy is slightly worse than the $k$NNG for the denser case.
It is intuitive because the v$k$NNG and $k$NNG become very similar in this case.

In contrast, the $\epsilon$NG maintains higher accuracy from middle to high connectivities. This suggests that $epsilon$NG keeps noise events disconnected. 

In summary, v$k$NNG shows high accuracy with fewer edges, while the $\epsilon$NG generally needs more edges to achieve the same accuracy. However, finding optimal $k_i$ values for v$k$NNG could be challenging, whereas the best $\epsilon$ for the $\epsilon$NG can be determined easily (see Section~\ref{sub:estimate_e}).

\begin{figure}[t]
\centering
\includegraphics[width = 0.9\linewidth]{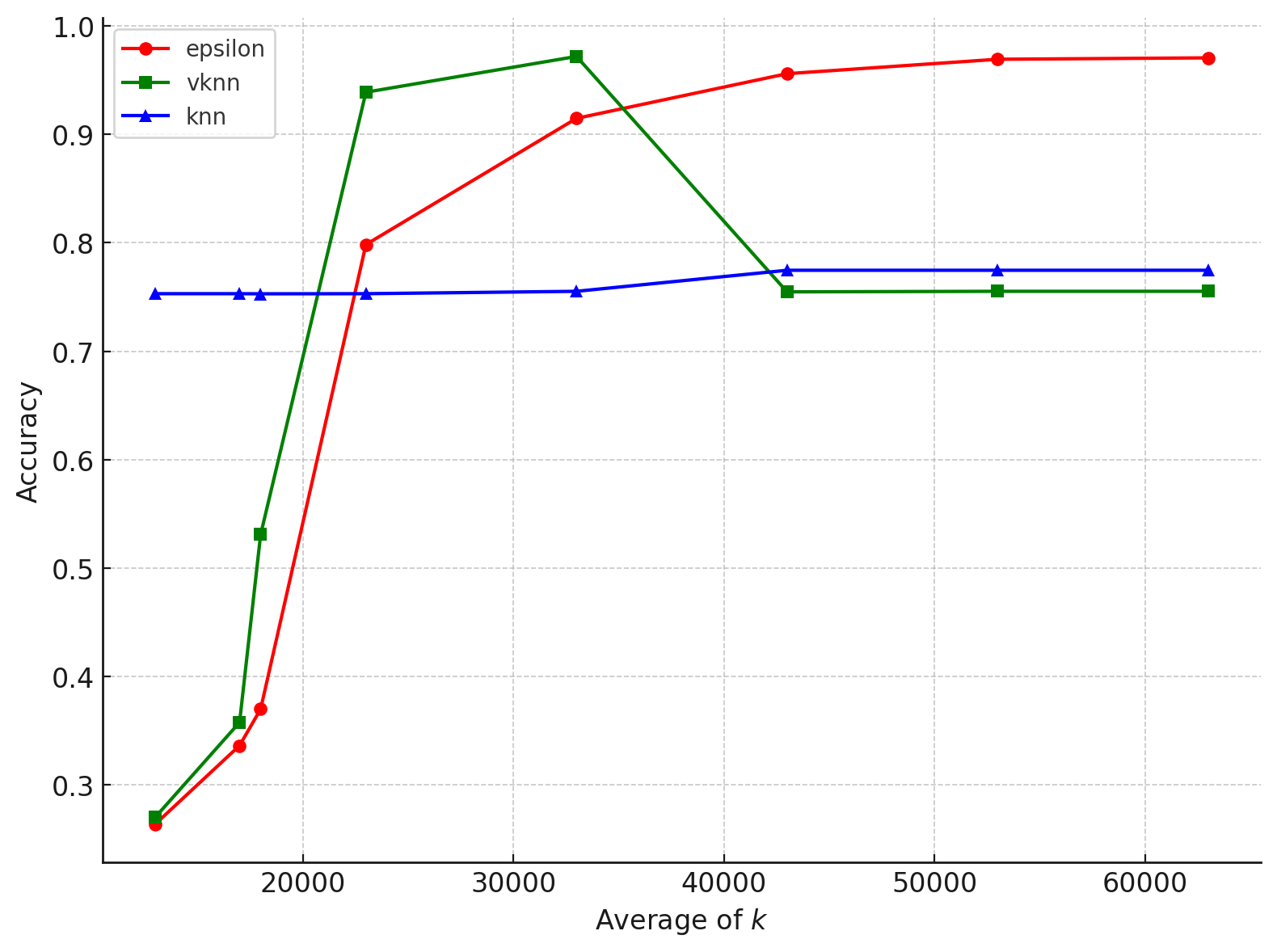} 
\caption{Accuracy of the proposed method for three graph constructions against varying connectivities.}
\label{fig:graph_copmarison}
\end{figure}

\section*{ACKNOWLEDGMENT}
The preferred spelling of the word ``acknowledgment'' in
American English is without an ``e'' after the ``g.'' Use the
singular heading even if you have many acknowledgments.
Avoid expressions such as ``One of us (S.B.A.) would like
to thank . . . .'' Instead, write ``F. A. Author thanks . . . .'' In
most cases, sponsor and financial support acknowledgments
are placed in the unnumbered footnote on the first page, not
here.

\bibliographystyle{IEEEtran}
\bibliography{bib}

\begin{IEEEbiographynophoto}
{SECOND B. AUTHOR,} photograph and biography not available at the time
of publication.
\end{IEEEbiographynophoto}

\begin{IEEEbiographynophoto}
{THIRD C. AUTHOR JR.}~(Member, IEEE), photograph and biography not available
at the time of publication.
\end{IEEEbiographynophoto}

\vfill\pagebreak

\end{document}